\documentclass[journal]{IEEEtran}

\usepackage{etoolbox}
\usepackage{hyperref}
\usepackage{amsmath,amsfonts,amssymb,amsthm}
\usepackage{mathtools}

\usepackage{algorithm, algpseudocode}

\usepackage{array, multirow,booktabs,siunitx,threeparttable}
\usepackage[caption=false,font=normalsize,labelfont=sf,textfont=sf]{subfig}
\usepackage{graphicx}
\usepackage{stfloats}
\usepackage{enumitem}
\usepackage{todonotes}
\newcommand{\Oh}{\mathcal{O}}       % Big O notation (uses \mathcal{O})
\newcommand{\OmegaBig}{\Omega}     % Big Omega notation (uses standard \Omega)
\usepackage{textcomp}
\usepackage{url}
\usepackage{verbatim}
\usepackage{cite}
\usepackage{makecell}
\usepackage{hyperref}

\setcellgapes{3pt}

\newtheorem{theorem}{Theorem}

\newtheorem{lemma}{Lemma}

\begin{document}
\renewcommand{\qedsymbol}{}
\title{Flexible Coded Distributed Convolution Computing for Enhanced Straggler Resilience and Numerical Stability in Distributed CNNs}
\author{Shuo Tan, Rui Liu, Xuesong Han, Xianlei Long,~\IEEEmembership{Member,~IEEE}, Kai Wan,~\IEEEmembership{Member,~IEEE},\\ Linqi Song,~\IEEEmembership{Member,~IEEE}, and Yong Li,~\IEEEmembership{Member,~IEEE}% <-this % stops a space  % <-this % stops a space
\thanks{This work is supported by the Fundamental Research Funds for the Central Universities
under Grant 2024CDJGF-034.}
\thanks{Shuo Tan, Rui Liu, Xianlei Long, Yong Li are with the College of Computer
Science, Chongqing University, Chongqing 400044, China (e-mail: shuotan@cqu.edu.cn, liurui\_cs@cqu.edu.cn, xianlei.long@cqu.edu.cn, yongli@cqu.edu.cn).}
\thanks{Xuesong Han is with the School of Electrical Engineering, Korea University, Seoul 02841, South Korea (e-mail: xuesonghan@korea.ac.kr).}
\thanks{Kai Wan is with the School of Electronic Information and Communications, Huazhong University of Science and Technology, Wuhan 430074, China (e-mail: kai\_wan@hust.edu.cn).}
\thanks{Linqi Song is with City University of Hong Kong, Hong Kong and City University of Hong Kong Shenzhen Research Institute, Shenzhen, China (email: linqi.song@cityu.edu.hk).}}
\markboth{}%
{Shell \MakeLowercase{\textit{et al.}}: A Sample Article Using IEEEtran.cls for IEEE Journals}
\maketitle
\begin{abstract}
Deploying Convolutional Neural Networks (CNNs) on resource-constrained devices necessitates efficient management of computational resources, often via distributed environments susceptible to latency from straggler nodes. This paper introduces the Flexible Coded Distributed Convolution Computing (FCDCC) framework to enhance straggler resilience and numerical stability in distributed CNNs. We extend Coded Distributed Computing (CDC) with Circulant and Rotation Matrix Embedding (CRME) which was originally proposed for matrix multiplication to high-dimensional tensor convolution. For the proposed scheme, referred to as the Numerically Stable Coded Tensor Convolution (NSCTC) scheme, we also propose two new coded partitioning schemes: Adaptive-Padding Coded Partitioning (APCP) for the input tensor and Kernel-Channel Coded Partitioning (KCCP) for the filter tensor. These strategies enable linear decomposition of tensor convolutions and encoding them into CDC subtasks, combining model parallelism with coded redundancy for robust and efficient execution. Theoretical analysis identifies an optimal trade-off between communication and storage costs. Empirical results validate the framework's effectiveness in computational efficiency, straggler resilience, and scalability across various CNN architectures.
\end{abstract}
\section{Introduction}
\IEEEPARstart{I}{n} the rapidly evolving domain of distributed machine learning, Convolutional Neural Networks (CNNs) have become fundamental due to their exceptional capabilities in image feature extraction and classification \cite{shen2022high,li2021contexts}, as well as their versatility enabled by transfer learning \cite{raghu2019transfusion,yosinski2014transferable}. A significant trend in this field, particularly relevant to Internet of Things (IoT) applications, is the shift towards edge computing, where data processing is conducted directly on edge devices \cite{shi2016edge}. This paradigm reduces reliance on cloud resources, minimizing latency \cite{rabinia2024algorithms} and enhancing data privacy \cite{wang2025privacy}, thereby improving service quality.

Deploying CNNs in distributed systems, especially on resource-constrained IoT devices, poses significant challenges due to intensive computational requirements, particularly within convolutional layers (ConvLs). Convolution operations represent over 90\% of the Multiply-Accumulate operations (MACs) in mainstream CNN architectures, including AlexNet, VGGNet, GoogleNet and ResNet \cite{fan2021rethinking}, and account for more than 80\% of the computational time during inference \cite{hadidi2023creating}.

Collaborative inference across multiple devices has emerged as a viable approach to mitigate the computational burden on the single device and enhance CNN deployment efficiency \cite{du2020model}. However, inference latency is often significantly impacted by slow worker nodes, commonly referred to as \emph{stragglers}. These stragglers, arising from hardware heterogeneity and variable network conditions \cite{chen2019deep}, can lead to performance degradation and potential failures, particularly in IoT systems where data loss rates may exceed 70\% per layer \cite{hadidi2023creating}.

\emph{Coded Distributed Computing (CDC)} is a straggler-resilient paradigm that injects \emph{algorithmic redundancy} into parallel workloads executed on unreliable clusters.  
Assume a job can be decomposed into \(m\) independent subtasks.  
The master linearly encodes these subtasks with an \((n,m)\) error-correcting code and dispatches the resulting \(n>m\) coded subtasks to \(n\) workers.  
Because any subset of \(k\) returned results, \(m\le k<n\), suffices for decoding, the master could disregard the slowest \(n-k\) workers; the minimum such \(k\) is called the \emph{recovery threshold}.  
By adding coded redundancy for deterministic resilience, CDC provably shortens the expected completion time under heterogeneous latency and even node failures.
CDC has been successfully applied in Coded Matrix-Matrix Multiplication (CMMM) and other matrix multiplication-based algorithms \cite{lee2017speeding,yu2017polynomial,fahim2017optimal, yu2019lagrange,dutta2019optimal, jia2021cross, das2022coded,li2022flexible, wang2024addressing, makkonen2024general}, due to the ease of linear decomposition of these operations.

The application of CDC in large–scale deep-learning systems has attracted increasing research interest.  
Dutta \textit{et al.} pioneered a unified CDC framework for deep neural networks (DNNs) based on Generalised PolyDot codes \cite{dutta2018unified}, demonstrating that very large DNN models can be executed on unreliable hardware susceptible to soft errors. Subsequent studies \cite{ji2023coded,hadidi2023creating} have integrated CDC at the model-parallel level, focusing predominantly on fully connected layers (FCLs) that perform operations of the form $\mathbf{Y} = \mathbf{W}\mathbf{X}$, where CMMM-based schemes were used to address straggler effects through encoding of the weight matrix $\mathbf{W}$ and input matrix $\mathbf{X}$.

However, existing CMMM-based schemes are not directly extensible to convolutional neural networks (CNNs), given that ConvLs process high-order tensor computations rather than traditional matrix operations. For a tensor \(\mathcal T\in\mathbb R^{d_1\times\cdots\times d_N}\), CDC decomposition introduces multiple, mode-dependent splitting dimensions that demand purpose-built partition strategies. 
Moreover, numerical stability is a critical design axis for any practical CDC deployment \cite{sun2022surprising}.
Most CDC codes—including Polynomial \cite{yu2017polynomial}, MatDot \cite{fahim2017optimal}, PolyDot \cite{dutta2019optimal}, and Lagrange \cite{yu2019lagrange}—rely on real-valued evaluations of Vandermonde polynomials.The condition number of an \(n\times n\) real Vandermonde matrix increases exponentially with \(n\) \cite{gautschi1987lower}; empirical evidence shows that decoding often fails for \(n\ge 30\) workers because the matrix inversion becomes ill-conditioned \cite{subramaniam2019random,fahim2021numerically} which limits the scalability of the system.

To mitigate this limitation, Ramamoorthy~\textit{et~al.} strengthened the numerical stability of CMMM by integrating \emph{circulant–rotation matrix embeddings} (CRME) into polynomial codes~\cite{ramamoorthy2021numerically}. CRME—long utilized in numerical simulation~\cite{olson2014circulant} and signal processing~\cite{gray2006toeplitz}—leverages the well-conditioned nature of complex Vandermonde matrices whose evaluation points lie on the unit circle, while all arithmetic remains in the real field \(\mathbb{R}\).  Consequently, the recovery matrix stays well conditioned even when the system scales beyond \(n\ge 100\) workers, demonstrating CRME's applicability to real-valued CDC scenarios, including tensor convolution operations in CNNs.

Integrating CDC into CNNs poses significant challenges due to the complex and tightly coupled nature of ConvLs. These layers involve intricate interactions between three-dimensional input tensors (feature maps) and four-dimensional filter tensors (kernels), requiring careful management during tensor decomposition to preserve spatial continuity at the model parallelism level \cite{coates2013deep}. 
Existing research on CDC within this field was limited to one-dimensional vector convolutions \cite{dutta2017coded,zhou2022dynamic}, which were insufficient for CNN architectures.

Recent efforts, such as the Resilient, Secure, and Private Coded Convolution (RSPCC) scheme by Qiu \textit{et al.} \cite{qiu2024resilient}, have sought to broaden CDC's applicability to ConvLs by adapting the im2col-based algorithm \cite{li2016performance}, transforming tensor convolutions into matrix-vector multiplications. However, due to this specific pre-encoding transformation, RSPCC exhibits limited compatibility with alternative convolution algorithms, including FFT-based methods \cite{vasilache2015fast}, Winograd-based algorithms \cite{lavin2016fast}, and approximation methods \cite{denton2014exploiting}. Additionally, RSPCC requires each node to retain a complete filter, thereby increasing storage requirement. It also relies on finite field computations, which introduce numerical stability challenges in real-valued environments \cite{ramamoorthy2021numerically}.

Addressing these challenges is crucial for the effective deployment of CDC in large-scale distributed CNNs. New general tensor block encoding and decoding strategies are required to manage high-dimensional structures, ensuring numerical stability while optimizing the trade-off between communication and storage costs.

This paper introduces a Flexible Coded Distributed Convolution Computing (FCDCC) framework designed specifically for ConvLs in CNNs within distributed environments. The framework enhances model parallelism through a numerically stable coded computing scheme, improving both efficiency and robustness while minimizing computational and memory requirements. Moreover, our approach is universally applicable: only the encoder and decoder act at the tensor level, so each worker may execute \emph{any} black-box tensor convolution algorithm.

The primary contributions of this work are as follows:
\begin{itemize}
    \item \textbf{Numerically Stable Coded Tensor Convolution (NSCTC)}: 
    We extend the CRME‑based CDC scheme from matrix multiplication to high‑dimensional tensor convolution by introducing the first tensor–matrix multiplication operation for both encoding and decoding in CDC based on our knowledge. This approach leverages complex Vandermonde matrices computed over the real field to achieve numerical stability, achieving a maximum mean squared error (MSE) of \(10^{-27}\) for AlexNet's ConvLs in a distributed setting with 20 worker nodes. This represents the first CDC scheme specifically designed for high-dimensional tensor operations.
    \item \textbf{Adaptive-Padding Coded Partitioning (APCP)}:  
    We introduce an adaptive partitioning strategy that divides the input tensor along its spatial dimensions (height \(H\) or width \(W\)), based on kernel size (\(K_H, K_W\)) and stride \(s\), with the addition of coded redundancy. This reduces communication cost and workload per node compared to the single-node scheme while maintaining resilience against stragglers.
    \item \textbf{Kernel-Channel Coded Partitioning (KCCP)}:  
    We employ a partitioning approach for the filter tensor along the output channel dimension \(N\) and generate coded partitions to enhance resilience. This approach effectively reduces both the storage cost and the per-node workload, while simultaneously enhancing straggler resilience relative to the RSPCC scheme---which only incorporates coded redundancy in the inputs.
    \item \textbf{Framework Optimization}: The FCDCC framework is then analyzed to flexibly choose optimal partitioning parameters \(k_A\) and \(k_B\), balancing communication and storage costs while maintaining a fixed number of subtasks (where \(k_A k_B\) is constant).
    \item \textbf{Generality}: The framework is applicable to CNN libraries such as PyTorch and various CNN models, including LeNet, AlexNet, and VGGNet.
\end{itemize}

The remainder of the paper is organized as follows: Section II presents the system model, Section III introduces the NSCTC scheme, Section IV describes the FCDCC framework and cost optimization, Section V provides theoretical analysis and complexity evaluation, Section VI offers experimental validation, and Section VII concludes with future research directions.

\emph{Notations:} The imaginary unit is denoted by \(\mathrm{i} = \sqrt{-1}\). The set of integers modulo \( m \) is represented by \(\mathcal{Z}_m = \{0, 1, \ldots, m-1\}\), and the cardinality of a finite set \(\mathcal{I}\) is denoted by \(|\mathcal{I}|\). We employ various matrix and tensor operations, including the Kronecker product (\(\otimes\)) and tensor convolution (\(*\)). For a matrix \(\mathbf{M}\), \(\mathbf{M}(i, j)\) denotes its \((i, j)\)-th entry; \(\mathbf{M}_{i,j}\) represents its \((i, j)\)-th block sub-matrix; and \(\mathbf{M}_{i}\) denotes its \(i\)-th column block sub-matrix. For a tensor \(\mathbf{T} \in \mathbb{R}^{N_1 \times N_2 \times N_3}\), \((\mathbf{T})_{i,j,k}\) (or \(t_{i,j,k}\)) refers to its \((i, j, k)\)-th entry; similar notation applies to higher-dimensional tensors. We denote a $1 \times U_k$ tensor block list as \(\mathbf{T}' = [\mathbf{T}'_0,\, \mathbf{T}'_1,\, \ldots,\, \mathbf{T}'_{U_k - 1}]\) or \(\{\mathbf{T}'_i\}_{i=0}^{U_k - 1}\), where each tensor block \(\mathbf{T}'_i\) has identical dimensions.

\section{System Model}
\subsection{System Overview}
The proposed FCDCC framework adopts a master--worker architecture with one master and $n$ homogeneous workers, and it tolerates up to $\gamma = n - \delta$ stragglers (i.e., workers that are significantly slower or unresponsive).
The master applies the designated tensor-partitioning scheme to the input tensor $\mathbf{X}$ and the filter tensor $\mathbf{K}$, thereby splitting the original tensor-convolution task into $\delta$ independent node-level subtasks.
These $\delta$ raw subtasks are then encoded to inject redundancy and expanded into $n$ coded node-level subtasks, which are dispatched one-to-one to distinct workers.
Upon receiving results from any $\delta$ workers, the master immediately decodes to recover the output tensor~$\mathbf{Y}$.
\begin{figure*}[!t]
\centering
\includegraphics[width=\textwidth]{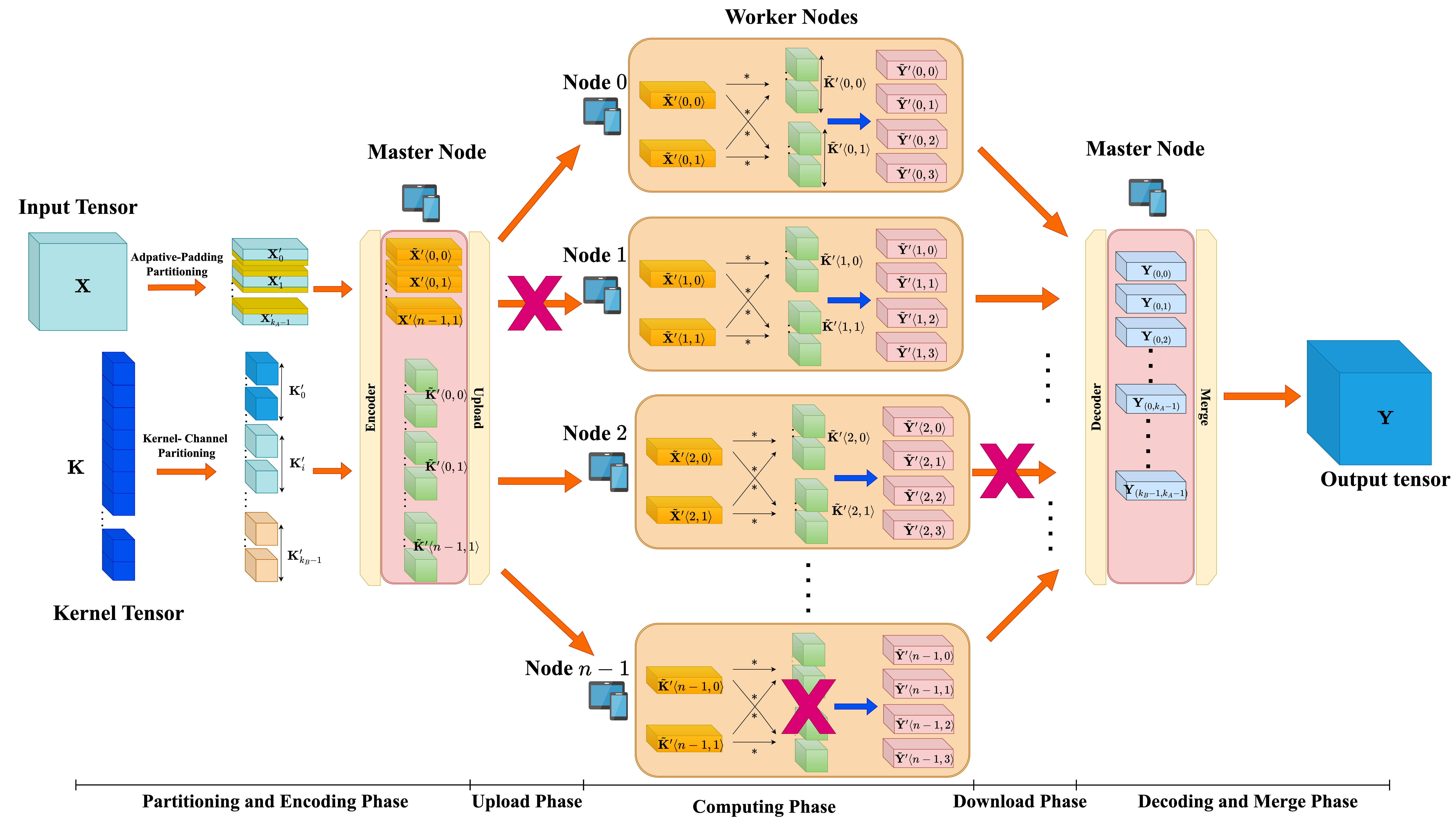}
\caption{The FCDCC framework demonstrating the main workflow and tensor definitions in a comprehensive 3-D representation. We assume 3 kinds of straggler problems in this diagram: Upload failures, computing failures, and download failures.}
\label{fig:framework}
\end{figure*}

Mostly used notation is summarized in Table \ref{tab:notation}. 
\begin{table}[!t]
  \centering
  \caption{\textsc{Notation Table}}
  \label{tab:notation}
  \renewcommand{\arraystretch}{1.2}
  \begin{tabular}{@{}l l@{}}
    \toprule
    \textbf{Notation} & \textbf{Description} \\
    \midrule
    $C$       & Number of input channels \\
    $N$       & Number of output channels \\
    $H, W$    & Height and width of the input tensor \\
    $K_H, K_W$ & Kernel (filter) height and width \\
    $p$       & Padding size \\
    $s$       & Stride length \\
    $H',W'$   & Height and width of the output tensor \\
    $\mathbf{X}$ & Input tensor $\in \mathbb{R}^{C\times(H+2p)\times(W+2p)}$ \\
    $\mathbf{K}$ & Filter tensor $\in \mathbb{R}^{N\times C\times K_H\times K_W}$ \\
    $\mathbf{Y}$ & Output tensor $\in \mathbb{R}^{N\times H' \times W'}$ \\
    $n$       & Number of worker nodes \\
    $\gamma$  & Straggler resilience capacity \\
    $\delta$  & Recovery threshold \\
    $k_A,k_B$ & Partition number for $\mathbf{X}$ and $\mathbf{K}$ \\
    $l$       & \begin{tabular}[t]{@{}l@{}}Number of coded partitions of each tensor ($\mathbf{X}$, $\mathbf{K}$) \\ assigned to a worker node.\end{tabular} \\
    \bottomrule
  \end{tabular}
\end{table}

Key parameters in this framework include:
\begin{itemize}
    \item \textbf{Recovery Threshold (\( \delta \))}: Define \(\delta = \left\lfloor \frac{k_A k_B}{\ell^2} \right\rfloor\), where \( k_A \) and \( k_B \) are partitioning parameters along the spatial (height or width) and output channel dimensions, respectively, and \( \ell \) is the subgroup parameter (\( \ell = 2 \) in CRME-based schemes; \( \ell = 1 \) in classical CDC schemes \cite{yu2017polynomial, fahim2017optimal, dutta2019optimal, yu2019lagrange}). The recovery threshold $\delta$ represents that by receiving the answers of any $\delta$ worker nodes, the server should recover the output tensor.

    \item \textbf{Storage Fraction Parameter (\( \Delta \))}: The ratio of tensor storage per worker node to the total tensor size, representing the (normalized) storage cost by each worker node.
\end{itemize}

The other key components and operations of the framework are described below.
\subsection{Tensor Operations}

The convolution operation \( * \) between the input tensor \(\mathbf{X}\) and the filter tensor \(\mathbf{K}\) results in the output tensor \(\mathbf{Y}\), where the spatial dimensions of \(\mathbf{Y}\), denoted as \(H'\) and \(W'\), are computed as \(H' = \left\lfloor \frac{H + 2p - K_H}{s} + 1 \right\rfloor\) and \(W' = \left\lfloor \frac{W + 2p - K_W}{s} + 1 \right\rfloor\), given stride \(s\) and padding \(p\) \cite{krizhevsky2012imagenet}. Our compute task is defined as:
\begin{equation}
\begin{aligned}
& \mathbf{Y}_{n, h, w}= (\mathbf{X} * \mathbf{K})_{n, h, w} \\
& = \sum_{c=0}^C \sum_{i=0}^{K_H-1} \sum_{j=0}^{K_W-1} (\mathbf{X})_{c, sh+i, s w+j}(\mathbf{K})_{n, c, i, j},
\end{aligned}
\end{equation}
for \(n \in \mathcal{Z}_N\), \(h \in \mathcal{Z}_{H'}\), \(w \in \mathcal{Z}_{W'}\), where stride \(s\) is applied. Here, \((\mathbf{X} * \mathbf{K})_{n,h,w}\) represents the value of the output tensor at position \((n, h, w)\), calculated by adding the element-wise products of the input tensor \(\mathbf{X}\) and the filter tensor \(\mathbf{K}\) across all input channels \(C\) and the dimensions of the kernel \((K_H, K_W)\).

\subsection{Distributed Computing Framework Operations}
The FCDCC framework operates through several key phases:
\subsubsection{Partitioning and Encoding Phase}
The master node executes the following tasks:

\begin{itemize}
    \item \textbf{Partitioning}: The input tensor \(\mathbf{X}\) is divided into \(k_A\) segments \(\{\mathbf{X}'_i\}_{i=0}^{k_A - 1}\) along the height dimension \(H\) (the same approach can be used for the width dimension \(W\)), while preserving the input channel dimension \(C\). Each partition is represented as \(\mathbf{X}'_i \in \mathbb{R}^{C \times \hat{H} \times (W + 2p)}\). This partitioning is performed using the Adaptive-Padding Partitioning scheme, which calculates the adaptive padded height \(\hat{H}\) and determines the adjusted starting index \(\hat{S}\) for each segment. 

    Simultaneously, the filter tensor \(\mathbf{K}\) is divided into \(k_B\) non-overlapping subtensors \(\{\mathbf{K}'_j\}_{j=0}^{k_B - 1}\) using the Kernel-Channel Partitioning scheme, along the output channel dimension \(N\), while maintaining the original kernel dimensions \((K_H, K_W)\) and input channel dimension \(C\). Each partition is represented as \(\mathbf{K}'_j \in \mathbb{R}^{\frac{N}{k_B} \times C \times K_H \times K_W}\).
        
    \item \textbf{Encoding}: The partitioned input tensors \(\{\mathbf{X}'_i\}_{i=0}^{k_A - 1}\) and filter tensors \(\{\mathbf{K}'_j\}_{j=0}^{k_B - 1}\) are encoded using Circulant and Rotation Matrix Embedding (CRME) matrices \(\mathbf{A} \in \mathbb{R}^{k_A \times n\ell}\) and \(\mathbf{B} \in \mathbb{R}^{k_B \times n\ell}\) through linear combinations, respectively, as follows:
    \small{
        \begin{align}  
            \tilde{\mathbf{X}}'_{\langle i,j \rangle} &= \sum_{\alpha=0}^{k_A - 1} \mathbf{A}(\alpha, i\ell + j) \mathbf{X}'_\alpha, \quad \text{for } i \in \mathcal{Z}_n, \; j \in \mathcal{Z}_\ell, \label{eq:encoding_X} \\
            \tilde{\mathbf{K}}'_{\langle i,j \rangle} &= \sum_{\beta=0}^{k_B - 1} \mathbf{B}(\beta, i\ell + j) \mathbf{K}'_\beta, \quad \text{for } i \in \mathcal{Z}_n, \; j \in \mathcal{Z}_\ell. \label{eq:encoding_K}
        \end{align}
    } 
    This encoding process generates \(n\ell\) encoded partitions \(\tilde{\mathbf{X}}'_{\langle i,j \rangle}\) and \(\tilde{\mathbf{K}}'_{\langle i,j \rangle}\) for each \(i \in \mathcal{Z}_n\) and \(j \in \mathcal{Z}_\ell\), where \(\langle i, j \rangle\) denotes the pair of indices corresponding to worker node \(i\) and tensor partition \(j\) assigned to that node. In the FCDCC framework, the filter tensor partitions \(\{\mathbf{K}'_j\}\) are typically encoded and distributed once during initialization, as CNN inference tasks generally utilize fixed models.
\end{itemize}

\subsubsection{Upload, Computing, and Download Phases}
\begin{itemize}
    \item \textbf{Upload}: Given that the encoded \(\ell\) filter tensors \(\{\tilde{\mathbf{K}}_{\langle i,j \rangle}\}_{j=0}^{\ell - 1}\) are typically pre-stored on each worker node \( i \), the master node only needs to transmit \( \ell \) distinct coded input tensor partitions \(\{\tilde{\mathbf{X}}_{\langle i,j \rangle}\}_{j=0}^{\ell - 1}\) to each of the \( n \) worker nodes.
    \item \textbf{Computing}: Upon receiving the encoded input tensors, each worker node \( i \) independently performs convolution operations between the \( \ell \) encoded input partitions and the \( \ell \) encoded filter partitions, resulting in \( \ell^2 \) encoded outputs \(\tilde{\mathbf{Y}}_{\langle i, j \rangle} \in \mathbb{R}^{\frac{N}{k_B} \times \frac{H'}{k_A} \times W'}\) for \( j \in \mathcal{Z}_{\ell^2} \). The worker node then concatenates these outputs along the channel dimension to obtain \(\tilde{\mathbf{Y}}_i \in \mathbb{R}^{\frac{\ell^2N}{k_B} \times \frac{H'}{k_A} \times W'}\).
    \item \textbf{Download}: Each worker node transmits its computed result \( \tilde{\mathbf{Y}}_i \) back to the master node for final reconstruction.
\end{itemize}
\subsubsection{Decoding and Merging Phase}
\begin{itemize}
    \item \textbf{Decoding}: Once results from at least \( \delta \) worker nodes are received, the master node employs a decoding matrix \( \mathbf{D} \) to reconstruct the original convolution results \(\mathbf{Y}_{(i,j)}\), for \( i \in \mathcal{Z}_{k_A} \) and \( j \in \mathcal{Z}_{k_B} \). The decoding matrix is configured based on the indices of the active worker nodes \( \mathcal{I} \) $(|\mathcal{I}| = \delta)$.
    \item \textbf{Merging}: The master node assembles the decoded $k_Ak_B$ tensor partitions to form the final output tensor \(\mathbf{Y} \in \mathbb{R}^{N \times H' \times W'}\). This involves concatenating the partitions along the spatial dimension (\( H' \) or \( W' \)) and the output channel dimension \( N \).
\end{itemize}

The entire framework is illustrated in Fig.~\ref{fig:framework}.
\subsection{Cost Definitions}\label{sec:cost_optimization}
In the FCDCC framework, the objective is to minimize the total cost per worker node, \(U_{k_A, k_B}\), under fixed parameters \((n, \delta, \gamma)\) and a constant partition product \(Q = k_A k_B\). The total cost comprises communication cost, computational cost, and storage cost, each evaluated based on their respective unit costs: \(\lambda_{\text{comm}}\) per tensor entry for communication, \(\lambda_{\text{comp}}\) per MAC operation for computation, and \(\lambda_{\text{store}}\) per tensor entry for storage. Specifically, the cost components are defined as follows:

\begin{itemize}
    \item \textbf{Upload Communication Cost (\(C_{\text{comm\_up}}\))}:  
    The upload communication cost is proportional to the number of input tensor entries transmitted to each node, denoted as \(V_{\text{comm\_up}}\):
    \begin{equation}
        C_{\text{comm\_up}} = \lambda_{\text{comm}} V_{\text{comm\_up}}.
    \end{equation}
    
    \item \textbf{Download Communication Cost (\(C_{\text{comm\_down}}\))}:  
    The download communication cost is proportional to the number of output tensor entries received from each node, denoted as \(V_{\text{comm\_down}}\):
    \begin{equation}
        C_{\text{comm\_down}} = \lambda_{\text{comm}} V_{\text{comm\_down}}.
    \end{equation}
    
    \item \textbf{Total Communication Cost (\(C_{\text{comm}}\))}:  
    The total communication cost is the sum of upload and download communication costs:
    \begin{equation}
    \begin{aligned}
    C_{\text{comm}} &= C_{\text{comm\_up}} + C_{\text{comm\_down}}\\
    &= \lambda_{\text{comm}} \left( V_{\text{comm\_up}} + V_{\text{comm\_down}} \right).
    \end{aligned}
    \end{equation}
    
    \item \textbf{Storage Cost (\(C_{\text{store}}\))}:  
    The storage cost is proportional to the number of filter tensor entries stored on each node, denoted as \(V_{\text{store}}\):
    \begin{equation}
        C_{\text{store}} = \lambda_{\text{store}} V_{\text{store}}.
    \end{equation}
    
    \item \textbf{Computational Cost (\(C_{\text{comp}}\))}:  
    The computational cost is proportional to the number of Multiply-Accumulate (MAC) operations required for tensor convolution per node, denoted as \(M_{\text{comp}}\):
    \begin{equation}
        C_{\text{comp}} = \lambda_{\text{comp}} M_{\text{comp}}.
    \end{equation}
\end{itemize}

Therefore, the total cost per worker node is given by:
\begin{equation}
    U_{k_A, k_B} = C_{\text{comm}} + C_{\text{comp}} + C_{\text{store}}.
\end{equation}

To determine the optimal partitioning strategy, we formulate the following optimization problem:
\begin{equation}
\begin{aligned}
(k_A^*, k_B^*) = \underset{k_A, k_B}{\arg\min} \quad & U_{k_A, k_B} \\
\text{subject to} \quad & k_A, k_B \in \mathcal{S}, \\
& k_A k_B = Q,
\end{aligned}
\label{eq:optimization_problem}
\end{equation}
where \( \mathcal{S} = \left\{ x \in \mathbb{Z}^+ \,\big|\, x \equiv 0 \ (\mathrm{mod}\ \ell) \text{ or } x = 1 \right\} \) is the set of permissible partitioning factors. The optimal values \( (k_A^*, k_B^*) \) will be determined in Section IV.

\section{Numerically Stable Coded Tensor Convolution}

To address the limitations of traditional CDC schemes in high-dimensional convolution operations, this section introduces two complementary schemes: (I) partitioning tensor lists into subgroups for parallel convolutions, and (II) utilizing tensor-matrix multiplication for encoding and decoding operations. The proposed NSCTC methodology first employs Scheme I to perform linear decomposition of the tensor computing task, and subsequently applies Scheme II to introduce redundancy and create coded subtasks. Upon reception of a predetermined number of completed subtasks, the entire computational task can be recovered with guaranteed reliability. These techniques collectively establish NSCTC as the foundational component of the FCDCC framework.

To optimize the computational efficiency of CDC schemes, we decompose the input and filter tensor lists into disjoint subgroups parameterized by the partition factor \(\ell \in \mathbb{N}^+\). We assume that both \(k_A\) and \(k_B\) are exact multiples of \(\ell\) (i.e., \(k_A \equiv 0 \pmod{\ell}\) and \(k_B \equiv 0 \pmod{\ell}\)), where \(k_A, k_B \in \mathbb{N}^+\) denote the respective sizes of the input and filter tensor lists. This ensures uniform subgroup partitioning with precisely \(\frac{k_A}{\ell}\) and \(\frac{k_B}{\ell}\) tensors per subgroup, respectively. Let \(\mathbf{T}_A \in \mathbb{R}^{k_A \times (C \times H \times W)}\) be a list of \(k_A\) 3D input tensors partitioned into \(\frac{k_A}{\ell}\) subgroups and \(\mathbf{T}_B \in \mathbb{R}^{ k_B \times (N \times C \times K_H \times K_W)}\) be a list of \(k_B\) 4D filter tensors partitioned into \(\frac{k_B}{\ell}\) subgroups, defined as follows:
\begin{align}
\mathbf{T}_A &= \left[\mathbf{T}_{A\langle 0,0\rangle}, \ldots, \mathbf{T}_{A\left\langle \frac{k_A}{\ell}-1,\ell-1\right\rangle}\right], \label{eq:T_A_partition} \\
\mathbf{T}_B &= \left[\mathbf{T}_{B\langle 0,0\rangle}, \ldots, \mathbf{T}_{B\left\langle \frac{k_B}{\ell}-1,\ell-1\right\rangle}\right], \label{eq:T_B_partition}
\end{align}
where \(\mathbf{T}_{A\langle i,j\rangle} \in \mathbb{R}^{C \times H \times W}\) and \(\mathbf{T}_{B\langle i,j\rangle} \in \mathbb{R}^{N \times C \times K_H \times K_W}\).
The output tensor list $\mathbf{T}_C$ is defined, and the convolution operation $*$ for tensor lists $\mathbf{T}_A$ and $\mathbf{T}_B$ is given by:
\begin{equation}
\begin{aligned}
\mathbf{T}_C &= \mathbf{T}_A * \mathbf{T}_B \\
&= \big[\mathbf{T}_{A\langle 0,0\rangle}, \ldots, \mathbf{T}_{A\left\langle \lfloor k_A-1/\ell \rfloor, \ell-1\right\rangle}\big] * \\
& \quad \big[\mathbf{T}_{B\langle 0,0\rangle}, \ldots, \mathbf{T}_{B\left\langle \lfloor k_B-1/\ell \rfloor, \ell-1\right\rangle}\big] \\
&= \left[\mathbf{T}_{C\langle 0,0\rangle}, \mathbf{T}_{C\langle 0,1\rangle}, \ldots, \mathbf{T}_{C\langle k_A-1, k_B-1\rangle}\right],
\end{aligned}
\end{equation}
where
\begin{equation}
\mathbf{T}_{C\langle i,j\rangle} = \mathbf{T}_{A\langle \lfloor i/\ell \rfloor, i \bmod \ell\rangle} * \mathbf{T}_{B\langle \lfloor j/\ell \rfloor, j \bmod \ell\rangle}.
\end{equation}

To integrate CRME, we define $\ell=2$ and use a $2 \times 2$ rotation matrix $\mathbf{R}_\theta$ :
\begin{equation}\label{eq:rotationmatrix}
    \mathbf{R}_\theta = 
    \begin{bmatrix}
        \cos \theta & -\sin \theta \\
        \sin \theta & \cos \theta
    \end{bmatrix}.
\end{equation}
We define the rotation angle \(\theta = 2\pi / q\), where \(q \geq n\) is an odd integer, with \(n\) representing the number of worker nodes, and impose the constraint \(\delta = \frac{k_A k_B}{4} \leq n\). This configuration ensures that each worker node stores a fraction \(\Delta_A = 2/k_A\) of \(\mathbf{T}_A\) and a fraction \(\Delta_B = 2/k_B\) of \(\mathbf{T}_B\), respectively. The \((i, j)\)-th block of the encoding matrices is defined by specific powers of the rotation matrix \(\mathbf{R}_\theta\) as follows:
\begin{equation}  
\begin{aligned}  
\mathbf{G}_{i, j}^A &= \mathbf{R}_\theta^{j i}, \quad \text{for } i \in \mathcal{Z}_{\frac{k_A}{2}}, j \in \mathcal{Z}_n, \\
\mathbf{G}_{i, j}^B &= \mathbf{R}_\theta^{(j \frac{k_A}{2}) i}, \quad \text{for } i \in \mathcal{Z}_{\frac{k_B}{2}}, j \in \mathcal{Z}_n.
\end{aligned}  
\end{equation}

Thus, the encoding matrices \(\mathbf{G}^A\) and \(\mathbf{G}^B\) are defined as:
\begin{equation}  
\begin{aligned}  
\label{eq:GAB}
\mathbf{G}^A &= [\mathbf{G}^A_0 \,\, \mathbf{G}^A_1 \,\, \ldots \,\, \mathbf{G}^A_{2n-1}] \\
&= 
\left[\begin{array}{cccc}  
\mathbf{I} & \mathbf{I} & \cdots & \mathbf{I} \\
\mathbf{I} & \mathbf{R}_\theta & \cdots & \mathbf{R}_\theta^{n-1} \\
\vdots & \vdots & \ddots & \vdots \\
\mathbf{I} & \mathbf{R}_\theta^{\frac{k_A}{2}-1} & \cdots & \mathbf{R}_\theta^{(n-1)\left(\frac{k_A}{2}-1\right)}
\end{array}\right], \\[10pt]
\mathbf{G}^B &= [\mathbf{G}^B_0 \,\, \mathbf{G}^B_1 \,\, \ldots \,\, \mathbf{G}^B_{2n-1}] \\
&=
\left[\begin{array}{cccc}  
\mathbf{I} & \mathbf{I} & \cdots & \mathbf{I} \\
\mathbf{I} & \mathbf{R}_\theta^{\frac{k_A}{2}} & \cdots & \mathbf{R}_\theta^{(n-1)\frac{k_A}{2}} \\
\vdots & \vdots & \ddots & \vdots \\
\mathbf{I} & \mathbf{R}_\theta^{(\frac{k_B}{2}-1)\frac{k_A}{2}} & \cdots & \mathbf{R}_\theta^{(n-1)(\frac{k_B}{2}-1)\frac{k_A}{2}}
\end{array}\right].
\end{aligned}  
\end{equation}
We define the multiplication of a $1 \times U_k$ tensor block list $\mathbf{T}$ and a $U_k \times U_n$ matrix $\mathbf{M}$ as follows, represented by the $\cdot$ operator:
\begin{equation}  
\begin{aligned}  
&\mathbf{T} \cdot \mathbf{M} =   
\begin{bmatrix}  
\mathbf{T}_0,\ldots, \mathbf{T}_{U_k-1}  
\end{bmatrix} \cdot   
\begin{bmatrix}  
m_{00} & \cdots & m_{0,U_n-1} \\   
\vdots & \ddots & \vdots \\
m_{U_k-1,0} & \cdots & m_{U_k-1,U_n-1}  
\end{bmatrix} \\[1em]  
&= \left[  
\sum_{i=0}^{U_k-1} m_{i,0} \mathbf{T}_i, \ldots, \sum_{i=0}^{U_k-1} m_{i,U_n-1} \mathbf{T}_i  
\right] = \tilde{\mathbf{T}}.  
\end{aligned}  
\end{equation}
Note that $\tilde{\mathbf{T}}$ is a $1 \times U_n$ tensor block list.

For the \(i\)-th worker node, considering \(\ell = 2 \), the encoded matrices are
\begin{equation}  
\begin{aligned}  
\tilde{\mathbf{T}}_{A\langle i, j\rangle} &= \mathbf{T}_{A} \cdot \mathbf{G}^A_{2i+j} \\
&= \mathbf{T}_{A\langle 0,j\rangle} + \sum_{k=1}^{\frac{k_A}{2}-1} \sum_{l=0}^1 \mathbf{R}_\theta^{ik}(l,j) \mathbf{T}_{A\langle k,l\rangle},\\
\tilde{\mathbf{T}}_{B\langle i, j\rangle} &= \mathbf{T}_{B} \cdot \mathbf{G}^B_{2i+j} \\
&= \mathbf{T}_{B\langle 0,j\rangle} + \sum_{k=1}^{\frac{k_B}{2}-1} \sum_{l=0}^1 \mathbf{R}_\theta^{i(\frac{k_A}{2})k}(l,j) \mathbf{T}_{B\langle k,l\rangle},
\end{aligned}  
\end{equation}
for \(j \in \mathcal{Z}_2 \) and \(i \in \mathcal{Z}_n\).
\begin{algorithm}[H]  
\caption{Numerically Stable Coded Tensor Convolution (NSCTC)}  
\label{alg:coded_tensor_conv}  
\begin{algorithmic}[1]  
\renewcommand{\algorithmicrequire}{\textbf{Input:}}  
\renewcommand{\algorithmicensure}{\textbf{Output:}}  
\Require Tensor lists $\mathbf{T}_A$, $\mathbf{T}_B$; number of worker nodes $n$;   
        $\ell = 2$; $q = \text{Nextodd}(n)$; $\theta = 2\pi / q$  
\Ensure Output tensor list $\mathbf{T}_C$  

\Statex \textbf{Initialization:}  
\State Construct rotation matrix $\mathbf{R}_\theta$ using \eqref{eq:rotationmatrix} 
\State Construct encoding matrices $\mathbf{G}^A$ and $\mathbf{G}^B$ using \eqref{eq:GAB}
\Statex \textbf{Tensor Encoding:}  
\For{$i \gets 0$ \textbf{to} $n-1$}  
    \For{$j \gets 0$ \textbf{to} $1$}  
        \State $\begin{aligned}  
            \tilde{\mathbf{T}}_{A\langle i, j\rangle} &\gets \mathbf{T}_{A\langle 0,j\rangle} + \sum_{k,l} \mathbf{R}_\theta^{ik}(l,j) \mathbf{T}_{A\langle k,l\rangle} \\
            \tilde{\mathbf{T}}_{B\langle i, j\rangle} &\gets \mathbf{T}_{B\langle 0,j\rangle} + \sum_{k,l} \mathbf{R}_\theta^{i(\frac{k_A}{2})k}(l,j) \mathbf{T}_{B\langle k,l\rangle}  
        \end{aligned}$
    \EndFor  
\EndFor  

\Statex \textbf{Parallel Convolution Computation:}  
\ForAll{$i \gets 0$ \textbf{to} $n - 1$} \textbf{in parallel on worker node $i$}
    \For{$\ell_1 \gets 0$ \textbf{to} $1$}  
        \For{$\ell_2 \gets 0$ \textbf{to} $1$}  
            \State Compute $\tilde{\mathbf{T}}_{A\langle i,\ell_1\rangle} * \tilde{\mathbf{T}}_{B\langle i,\ell_2\rangle}$  
        \EndFor  
    \EndFor  
\EndFor  

\Statex \textbf{Result Collection and Decoding:}  
\State Master node constructs tensor list $\tilde{\mathbf{T}}_C$ from $\delta$ worker nodes  
\State Construct matrix $\mathbf{G}^E$ using the results from $\delta$ worker nodes  
\State $\mathbf{T}_C \gets \tilde{\mathbf{T}}_C \cdot (\mathbf{G}^E)^{-1}$  
\State \Return $\mathbf{T}_C$  
\end{algorithmic}  
\end{algorithm}
The $i$-th worker node computes all pairwise convolutions between $\tilde{\mathbf{T}}_{A\langle i, j\rangle}$ and $\tilde{\mathbf{T}}_{B\langle i, j\rangle}$ for $j \in \mathcal{Z}_2$. These computations can be represented using the Kronecker product (\(\otimes\)). Let $\mathbf{T}_C$ denote a $1 \times k_A k_B$ tensor block list containing all pairwise tensor convolutions. The computation of \(\tilde{\mathbf{T}}_{A\langle i, \ell_1 \rangle} * \tilde{\mathbf{T}}_{B\langle i, \ell_2 \rangle}\) is given by:
\begin{equation}
    \begin{aligned}
        &\tilde{\mathbf{T}}_{A\langle i, \ell_1 \rangle} * \tilde{\mathbf{T}}_{B\langle i, \ell_2 \rangle} = (\mathbf{T}_{A} \cdot \mathbf{G}^A_{2i+\ell_1}) * (\mathbf{T}_{B} \cdot \mathbf{G}^B_{2i+\ell_2}) \\
        &= (\mathbf{T}_{A} * \mathbf{T}_{B}) \cdot (\mathbf{G}^A_{2i+\ell_1} \otimes \mathbf{G}^B_{2i+\ell_2}) \\
        &= \mathbf{T}_C \cdot \scalebox{0.93}{$\left(
        \begin{bmatrix}  
            \mathbf{I}(0, \ell_1) \\
            \mathbf{I}(1, \ell_1) \\
            \mathbf{R}_\theta^i(0, \ell_1) \\
            \mathbf{R}_\theta^i(1, \ell_1) \\
            \vdots \\
            \mathbf{R}_\theta^{i(\frac{k_A}{2}-1)}(0, \ell_1) \\
            \mathbf{R}_\theta^{i(\frac{k_A}{2}-1)}(1, \ell_1)  
        \end{bmatrix} 
        \otimes 
        \begin{bmatrix}  
            \mathbf{I}(0, \ell_2) \\
            \mathbf{I}(1, \ell_2) \\
            \mathbf{R}_\theta^i(0, \ell_2) \\
            \mathbf{R}_\theta^i(1, \ell_2) \\
            \vdots \\
            \mathbf{R}_\theta^{i(\frac{k_B}{2}-1)\frac{k_A}{2}}(0, \ell_2) \\
            \mathbf{R}_\theta^{i(\frac{k_B}{2}-1)\frac{k_A}{2}}(1, \ell_2)
        \end{bmatrix}
        \right)$}
    \end{aligned}.
\end{equation}
Then the output tensor block list $\tilde{\mathbf{T}}_{C_i}$ in the $i$-th worker node can be denoted as:
\begin{equation}  
\begin{aligned}
&\tilde{\mathbf{T}}_{C_i} = \left[\tilde{\mathbf{T}}_{A\langle i, 0 \rangle}, \tilde{\mathbf{T}}_{A\langle i, 1 \rangle}\right] * \left[\tilde{\mathbf{T}}_{B\langle i, 0 \rangle}, \tilde{\mathbf{T}}_{B\langle i, 0 \rangle}\right]\\
&= \mathbf{T}_{C} \cdot \scalebox{0.9}{$\left[ 
\mathbf{G}^A_{2i} \otimes \mathbf{G}^B_{2i}|\mathbf{G}^A_{2i} \otimes \mathbf{G}^B_{2i+1} | 
\mathbf{G}^A_{2i+1} \otimes \mathbf{G}^B_{2i} |\mathbf{G}^A_{2i+1} \otimes \mathbf{G}^B_{2i+1}
\right]$} \\
&= \mathbf{T}_{C} \cdot \left( 
\left[ \mathbf{G}^A_{2i} \,\, \mathbf{G}^A_{2i+1} \right] \otimes \left[ \mathbf{G}^B_{2i} \,\,\mathbf{G}^B_{2i+1} \right] 
\right) \\
&= \mathbf{T}_{C} \cdot 
\begin{bmatrix}  
\mathbf{I} \\
\mathbf{R}_\theta^i \\
\vdots \\
\mathbf{R}_\theta^{i(\frac{k_A}{2}-1)}
\end{bmatrix} 
\otimes 
\begin{bmatrix}  
\mathbf{I} \\
\mathbf{R}_\theta^{2i} \\
\vdots \\
\mathbf{R}_\theta^{i(\frac{k_B}{2}-1)\frac{k_A}{2}}
\end{bmatrix}.  
\end{aligned}
\end{equation}

Suppose that $\delta$ different worker nodes $i$ have finished their work. The master node obtains a $1 \times 4\delta$ tensor block list $\tilde{\mathbf{T}}_C$ as follows, considering $\tilde{\mathbf{T}}_{C_i}$ is a $1 \times 4$ tensor block list, for $i \in \mathcal{Z}_\delta$:
\begin{equation}  
\begin{aligned}  
\tilde{\mathbf{T}}_C &= [\tilde{\mathbf{T}}_{C_0}, \tilde{\mathbf{T}}_{C_1}, \ldots, \tilde{\mathbf{T}}_{C_\delta}] \\
&= \mathbf{T}_C \cdot \left(  
    \begin{bmatrix}  
        \mathbf{I} \\
        \mathbf{R}_\theta^{i_0} \\
        \vdots \\
        \mathbf{R}_\theta^{i_0\left(\frac{k_A}{2}-1\right)}  
    \end{bmatrix} \otimes  
    \begin{bmatrix}  
        \mathbf{I} \\
        \mathbf{R}_\theta^{2i_0} \\
        \vdots \\
        \mathbf{R}_\theta^{i_0\left(\frac{k_B}{2}-1\right)\frac{k_A}{2}}  
    \end{bmatrix}  
    \right. \\
& \left. \;\cdots \;  
    \begin{bmatrix}  
        \mathbf{I} \\
        \mathbf{R}_\theta^{i_\delta} \\
        \vdots \\
        \mathbf{R}_\theta^{i_\delta\left(\frac{k_A}{2}-1\right)}  
    \end{bmatrix} \otimes  
    \begin{bmatrix}  
        \mathbf{I} \\
        \mathbf{R}_\theta^{2i_\delta} \\
        \vdots \\
        \mathbf{R}_\theta^{i_\delta\left(\frac{k_B}{2}-1\right)\frac{k_A}{2}}  
    \end{bmatrix}  
\right)\\
&= \mathbf{T}_C \cdot \left(   
    \left[ \mathbf{G}^A_{2i_0} \; \mathbf{G}^A_{2i_0+1} \right] \otimes \left[ \mathbf{G}^B_{2i_0} \; \mathbf{G}^B_{2i_0+1}\right] \;\cdots \; \right. \\
&\left. \quad \left[ \mathbf{G}^A_{2i_\delta} \; \mathbf{G}^A_{2i_\delta+1} \right] \otimes \left[ \mathbf{G}^B_{2i_\delta} \;\mathbf{G}^B_{2i_\delta+1} \right]  
\right)\\
&= \mathbf{T}_C \cdot \mathbf{G}^E.
\end{aligned}
\end{equation}

Furthermore, the full-rank recovery matrix \(\mathbf{G}^E\) satisfies the polynomial bound \(\kappa(\mathbf{G}^E)=O\bigl(n^{\,n-\delta+5.5}\bigr)\)~\cite{ramamoorthy2021numerically} in our NSCTC scheme, while the Vandermonde-based recovery matrix of RSPCC \(\mathbf{G}^{E'}\) exhibits exponential growth \(\kappa(\mathbf{G}^{E'})>\Omega\bigl(e^{\,n-\delta}\bigr)\)~\cite{qiu2024resilient}, thus delivering a improvement in numerical stability.

Then, the coded tensor block list $\tilde{\mathbf{T}}_C$ is decodable by the following:
\begin{equation}  
\mathbf{T}_C = \tilde{\mathbf{T}}_C \cdot (\mathbf{G}^E)^{-1}.  
\end{equation}
The final output tensor list, denoted as \( \mathbf{T}_C \), can then be obtained. 

The NSCTC process within the CRME framework is detailed in Algorithm \ref{alg:coded_tensor_conv}. This algorithm enhances traditional CDC by enabling many-to-many high-dimensional linear tensor operations between tensor lists, ensuring numerical stability and improving computational efficiency. Its implementation within the FCDCC framework is discussed in Section IV.

\section{Flexible Coded Distributed Convolution Computing}
In this section, we embed the NSCTC scheme into the ConvLs of CNNs using APCP and KCCP schemes. By choosing specific dimensions, these partitioning methods decompose the convolution operations of input and filter tensors into encodable subtensors, which correspond to the tensor lists introduced in Section III. Furthermore, we propose an optimal partitioning scheme to achieve cost efficiency within the FCDCC framework.

\subsection{Adaptive-Padding Coded Partitioning (APCP)}
APCP employs Adaptive-Padding Partitioning to divide the input tensor into multiple overlapping subtensors, ensuring continuity. These subtensors are then encoded using the CRME scheme.
\subsubsection{Adaptive-Padding Partitioning}
To facilitate efficient distributed convolution operations in ConvLs, we partition the input tensor \(\mathbf{X} \in \mathbb{R}^{C \times (H+2p) \times (W+2p)}\) along the \(H\) dimension into \(k_A\) contiguous subtensors \(\{\mathbf{X}'_i\}_{i=0}^{k_A - 1}\), incorporating necessary overlaps to ensure continuity and validity of the convolution. The dimensions and index ranges of each subtensor \(\mathbf{X}'_i\) are precisely determined to align with the corresponding partitions of the output tensor. For simplicity, we assume that both the output height \(H'\) satisfy \(H' \equiv 0 \mod k_A\). If not, zero-padding is applied to \(\mathbf{X}\) to extend \(H'\) to the nearest multiple of \(k_A\), maintaining computational integrity and enabling efficient parallel processing.

Each subtensor \(\mathbf{X}'_i\) undergoes adaptive padding to a height \(\hat{H}\) (\(\hat{H} > \frac{H}{k_A}\)), calculated as:
\begin{equation}\label{CompH}
    \hat{H} = \left(\frac{H'}{k_A} - 1\right)s + K_H,
\end{equation}
where \(s\) is the stride and \(K_H\) is the kernel height. The starting index for each subtensor is adjusted according to:
\begin{equation}\label{CompS}
    \hat{S} = \frac{H'}{k_A} s.
\end{equation}

We define the tensor partitioning operation for a tensor \(\mathbf{T} \in \mathbb{R}^{N_1 \times N_2 \times N_3}\) as:
\begin{equation}\label{eq:tensor_partition}
\mathbf{T}' = \mathbf{T}[:, :, v:e],
\end{equation}
where \(v\) is the starting index (inclusive), \(e\) is the ending index (exclusive), and \(:\) denotes all indices along a dimension. This operation selects a contiguous subset along the third dimension, resulting in \(\mathbf{T}' \in \mathbb{R}^{N_1 \times N_2 \times (e - v)}\).

Applying this partitioning to \(\mathbf{X}\), the subtensors are defined as:
\begin{equation}\label{eq:P_A}
\mathbf{X}'_i = \mathbf{X}\left[:,\, i\hat{S} : i\hat{S} + \hat{H},\, :\right], \quad \text{for } i \in \mathcal{Z}_{k_A}.
\end{equation}
This partitions \(\mathbf{X}\) along the height dimension into \(k_A\) overlapping subtensors \(\mathbf{X}'_i\), each aligned with the corresponding output tensor partition.

Figure~\ref{fig:Partition} illustrates the partitioning process. In this example, an input tensor \(\mathbf{X} \in \mathbb{R}^{1 \times 10 \times 10}\) is convolved with a filter tensor \(\mathbf{K} \in \mathbb{R}^{1 \times 3 \times 3}\) using a stride \(s = 1\), where \(\hat{H} = 4\) and \(\hat{S} = 2\). Only the input tensor partitions corresponding to two output tensor blocks are shown.
\begin{figure}[htbp]
  \centering
  \includegraphics[width=0.8\columnwidth]{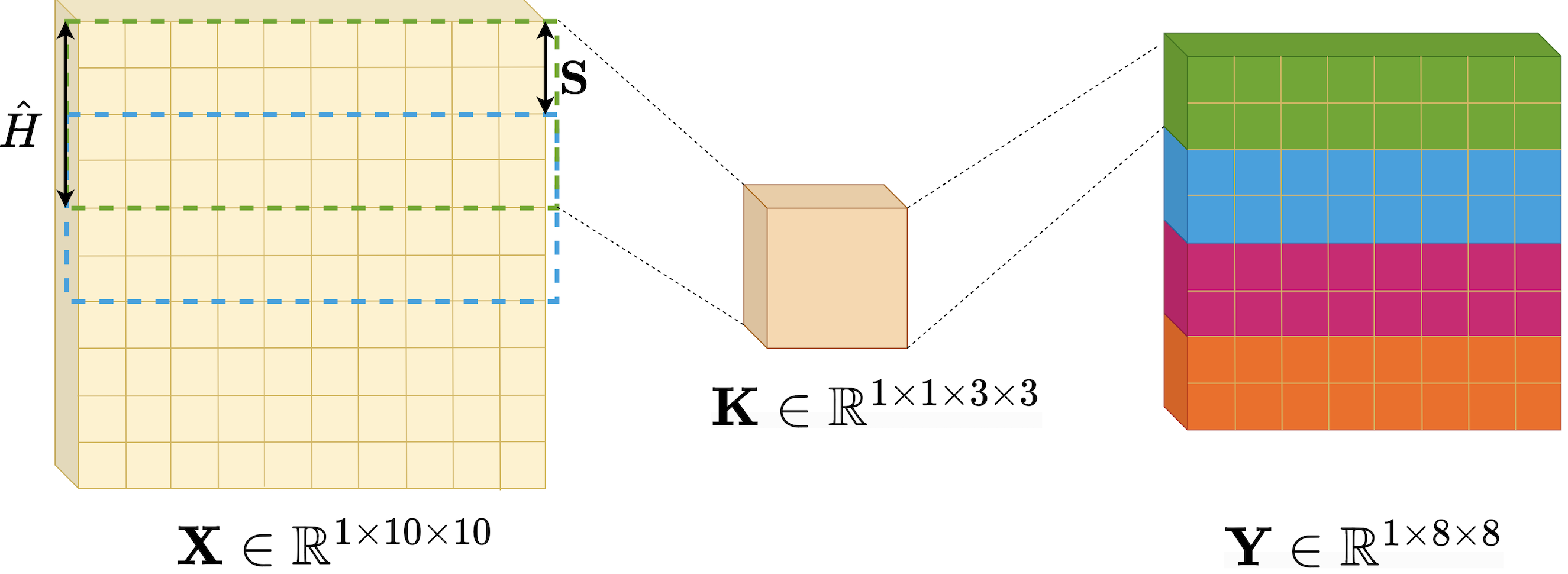}
  \caption{Illustration of the Adaptive-Padding Partitioning scheme. Colored blocks represent output tensor partitions derived from the convolution of the corresponding input tensor subtensors (dashed lines).}
  \label{fig:Partition}
\end{figure}

The complete list of input subtensors used for subsequent convolution operations is:
\begin{equation}
\mathbf{X}' = \left[\mathbf{X}'_0,\, \mathbf{X}'_1,\, \ldots,\, \mathbf{X}'_{k_A - 1}\right].
\end{equation}
\subsubsection{Input Tensor Partitions' Encoding Process}
Following Adaptive-Padding Partitioning, subtensors are encoded using CRME matrices. The encoding matrix \(\mathbf{A}\) is designed such that \(k_A\) denotes the number of original partitions, and \(2n\) represents the number of encoded partitions. The matrix \(\mathbf{A}\) is constructed using an angle \(\theta = \frac{2\pi}{q}\), where \(q \geq n\) is a odd integer. The \((i, j)\)-th blocks of the encoding matrices are then computed using rotation matrix $\mathbf{R}_\theta$:
\begin{align}\label{eq:encoding matrix A}
    \mathbf{A}_{i, j} = (\mathbf{R}_\theta)^{ji}, \quad &\text{for } i \in \mathcal{Z}_{\frac{k_A}{2}}, j \in \mathcal{Z}_{n}.
\end{align}

We introduce indexing parameters \(\alpha\) and \(\beta\), where \(\alpha = \lfloor i / 2 \rfloor\) and \(\beta = i \bmod 2\). Each subtensor \(\mathbf{X}'_{i}\) is then indexed as \(\mathbf{X}'_{\langle \alpha, \beta \rangle}\):
\begin{equation}\label{eq:X_GS}
    \mathbf{X}' = \left[ \mathbf{X}'_{\langle 0, 0\rangle}, \mathbf{X}'_{\langle 0, 1\rangle} \;\ldots \;\mathbf{X}'_{\langle \lfloor \frac{k_A}{2}\rfloor, 0\rangle}, \mathbf{X}'_{\langle \lfloor \frac{k_A}{2}\rfloor, 1\rangle} \right].
\end{equation}

The encoded tensor list \(\tilde{\mathbf{X}}'\) is then computed using the matrix \(\mathbf{A} \in \mathbb{R}^{k_A \times 2n}\):
\begin{equation}
\begin{aligned}
\label{eq:ENX}
\tilde{\mathbf{X}}' &= \mathbf{X}' \cdot \mathbf{A} \\
&= \left[ \tilde{\mathbf{X}}'_{\langle 0, 0\rangle}, \tilde{\mathbf{X}}'_{\langle 0, 1\rangle} \;\ldots \; \tilde{\mathbf{X}}'_{\langle n-1, 0\rangle}, \tilde{\mathbf{X}}'_{\langle n-1, 1\rangle} \right],
\end{aligned}
\end{equation}
where
\begin{equation}
\begin{aligned}
    \tilde{\mathbf{X}}'_{\langle i,j \rangle} &= \sum_{\alpha \in \mathcal{Z}_{\lfloor k_A / 2 \rfloor}} 
       \sum_{\beta \in \mathcal{Z}_{2}} 
       \mathbf{A}(2\alpha + \beta, 2i + j) \mathbf{X}'_{\langle \alpha, \beta \rangle}, \\
    &\quad \text{for } i \in \mathcal{Z}_{n}, \; j \in \mathcal{Z}_{2}.
\end{aligned}
\end{equation}

All encoding operations occur at the master node and the procedure of APCP scheme is presented in Algorithm \ref{alg:APCP}.
\begin{algorithm}[H]  
\caption{Adaptive-Padding Coded Partitioning (APCP)}  
\label{alg:APCP}  
\begin{algorithmic}[1]  
\renewcommand{\algorithmicrequire}{\textbf{Input:}}  
\renewcommand{\algorithmicensure}{\textbf{Output:}}  
\Require Input tensor $\mathbf{X} \in \mathbb{R}^{C \times (H + 2p) \times (W + 2p)}$, partition number $k_A$, stride $s$, kernel height $K_H$, worker node number $n$, output tensor height $H'$, $q = \text{Nextodd}(n)$  
\Ensure Encoded tensor list $\tilde{\mathbf{X}}'$  
\Statex \textbf{Adaptive-Padding Partitioning:}  
\State $\hat{H} \gets \left(\frac{H'}{k_A} - 1\right) \times s + K_H$  
\State $\hat{S} \gets \frac{H'}{k_A} \times s$  
\State Initialize $\mathbf{X}' \gets \left[\,\emptyset \,\right]$ 
\For{$i \gets 0$ \textbf{to} $k_A - 1$}  
    \State $\alpha \gets \lfloor i / 2 \rfloor$, $\beta \gets i \bmod 2$  
    \State $\mathbf{X}'_{\langle \alpha, \beta \rangle} \gets \mathbf{X}[:, i \hat{S} : \hat{H} + i \hat{S}, :]$  
    \State Append $\mathbf{X}'_{\langle \alpha, \beta \rangle}$ to $\mathbf{X}'$  
\EndFor  
\Statex \textbf{Encoding Input Tensor Partitions:}
\For{$i \gets 0$ \textbf{to} $k_A - 1$}
    \For{$j \gets 0$ \textbf{to} $n - 1$}
        \State Construct $\mathbf{A}_{i,j} = (\mathbf{R}_{2\pi/q})^{(ji)}$  
    \EndFor
\EndFor
\State Initialize $\tilde{\mathbf{X}}' \gets \left[\,\emptyset \,\right]$
\For{$i \gets 0$ \textbf{to} $n - 1$}  
    \For{$j \in \{0, 1\}$}  
        \State \small{$\tilde{\mathbf{X}}'_{\langle i, j \rangle} \gets \sum_{\alpha = 0}^{\lfloor k_A / 2 \rfloor - 1} \sum_{\beta = 0}^1 \mathbf{A}(2\alpha + \beta, 2i + j) \mathbf{X}'_{\langle \alpha, \beta \rangle}$}  
        \State Append $\tilde{\mathbf{X}}'_{\langle i, j \rangle}$ to $\tilde{\mathbf{X}}'$  
    \EndFor  
\EndFor  
\State \Return $\tilde{\mathbf{X}}'$  
\end{algorithmic}  
\end{algorithm}
\subsection{Kernel-Channel Coded Partitioning (KCCP)}
KCCP employs non-overlapping partitioning along the output channel dimension of the filter tensor, followed by encoding using the CRME scheme.
\subsubsection{Kernel-Channel Partitioning}
In the Kernel-Channel Coded Partitioning (KCCP) scheme, the master node partitions the filter tensor \(\mathbf{K} \in \mathbb{R}^{N \times C \times K_H \times K_W}\) into \(k_B\) subtensors along the output channel dimension \(N\). Specifically, each subtensor \(\mathbf{K}'_i \in \mathbb{R}^{\frac{N}{k_B} \times C \times K_H \times K_W}\) is defined as:
\begin{equation}\label{eq:P_B}
    \mathbf{K}'_i = \mathbf{K}\left[ i \tfrac{N}{k_B} : (i+1) \tfrac{N}{k_B},\; :,\; :,\; : \right], 
    \quad \text{for } i \in \mathcal{Z}_{k_B},
\end{equation}
where \( K_H \) and \( K_W \) are the kernel height and width, and \( C \) is the number of input channels. Each subtensor \(\mathbf{K}'_i\) retains the original kernel dimensions and input channels, encompassing a distinct set of output channels. This partitioning ensures that the structural integrity of the filter tensor is maintained while distributing the computational load across the worker nodes.
\subsubsection{Filter Tensor Partitions' Encoding Process}
Consistent with the structure of \(\mathbf{A}\), the \((i,j)\)-th block of encoding matrix \(\mathbf{B}\) within CRME is defined as:
\begin{align}
\mathbf{B}_{i,j} = (\mathbf{R}_\theta)^{(j \frac{k_A}{2}) i}, \quad \text{for } i \in \mathcal{Z}_{\frac{k_B}{2}}, \, j \in \mathcal{Z}_{n},
\end{align}

The partitioned filter tensor list \(\mathbf{K}'\) is then organized similarly to \eqref{eq:X_GS}:
\begin{equation}\label{eq:K_GS}
\begin{aligned}
    \mathbf{K}' &= \left[\mathbf{K}'_{\langle 0,0\rangle}, \mathbf{K}'_{\langle 0,1\rangle} \;\ldots \; \mathbf{K}'_{\langle \lfloor \frac{k_B}{2}\rfloor,0\rangle}, \mathbf{K}'_{\langle \lfloor \frac{k_B}{2}\rfloor,1\rangle}\right].
\end{aligned}
\end{equation}

The encoded tensor list \(\tilde{\mathbf{K}}'\) is then computed using the matrix \(\mathbf{B}\) as follows:
\begin{equation}\label{eq:ENK}
\begin{aligned}
\tilde{\mathbf{K}}' &= \mathbf{K}' \cdot \mathbf{B}\\
&= \left[ \tilde{\mathbf{K}}'_{\langle 0, 0\rangle}, \tilde{\mathbf{K}}'_{\langle 0, 1\rangle} \;\ldots \; \tilde{\mathbf{K}}'_{\langle n-1, 0\rangle}, \tilde{\mathbf{K}}'_{\langle n-1, 1\rangle} \right],
\end{aligned}
\end{equation}
where
\begin{equation}
\begin{aligned}
    \tilde{\mathbf{K}}'_{\langle i,j \rangle} &= \sum_{\alpha \in \mathcal{Z}_{\lfloor k_B / 2 \rfloor}} 
       \sum_{\beta \in \mathcal{Z}_{2}} 
       \mathbf{B}(2\alpha + \beta, 2i + j) \mathbf{K}'_{\langle \alpha, \beta \rangle}, \\
    &\quad \text{for } i \in \mathcal{Z}_{n}, \; j \in \mathcal{Z}_{2}.
\end{aligned}
\end{equation}
All encoding operations are completed at the master node. The procedure of KCCP scheme is presented in Algorithm 3.
\begin{algorithm}[H]  
\caption{Kernel-Channel Coded Partitioning (KCCP)}  
\label{alg:KCCP}  
\begin{algorithmic}[1]  
\renewcommand{\algorithmicrequire}{\textbf{Input:}}  
\renewcommand{\algorithmicensure}{\textbf{Output:}}  
\Require Filter Tensor $\mathbf{K} \in \mathbb{R}^{N \times C \times K_H \times K_W}$, partition number $k_B$, worker node number $n$, $q = \text{Nextodd}(n)$  
\Ensure Encoded filter tensor list $\tilde{\mathbf{K}}$  

\Statex \textbf{Kernel-Channel Partitioning:}  
\State Initialize $\mathbf{K}'\gets \left[\,\emptyset \,\right]$
\For{$i \gets 0$ \textbf{to} $k_B - 1$}  
    \State $\alpha \gets \lfloor i / 2 \rfloor$, $\beta \gets i \bmod 2$  
    \State $\mathbf{K}'_{\langle \alpha, \beta \rangle} \gets \mathbf{K}[i \cdot \frac{N}{k_B} : (i+1) \cdot \frac{N}{k_B}, :, :, :]$  
    \State Append $\mathbf{K}'_{\langle \alpha, \beta \rangle}$ to $\mathbf{K}'$  
\EndFor  

\Statex \textbf{Encoding Filter Tensor Partitions:}
\For{$i \gets 0$ \textbf{to} $k_B - 1$}
    \For{$j \gets 0$ \textbf{to} $n - 1$}
        \State Construct $\mathbf{B}_{i,j} = (\mathbf{R}_{2\pi/q})^{(j\frac{k_A}{2})i}$
    \EndFor  
\EndFor 
\State Initialize $\tilde{\mathbf{K}}\gets \left[\,\emptyset \,\right]$
\For{$i \gets 0$ \textbf{to} $n - 1$}  
    \For{$j \in \{0, 1\}$}  
        \State \small{$\tilde{\mathbf{K}}_{\langle i, j \rangle} \gets \sum_{\alpha = 0}^{\lfloor k_B / 2 \rfloor - 1} \sum_{\beta = 0}^1 \mathbf{B}(2\alpha + \beta, 2i + j) \mathbf{K}'_{\langle \alpha, \beta \rangle}$}  
        \State Append $\tilde{\mathbf{K}}_{\langle i, j \rangle}$ to $\tilde{\mathbf{K}}$  
    \EndFor  
\EndFor  

\State \Return $\tilde{\mathbf{K}}$  
\end{algorithmic}  
\end{algorithm}

\subsection[Upload and Computing Phase]{Upload and Computing Phase}

Different from classical CDC, FCDCC's upload and computing phase consists of these three steps:

\subsubsection{Pairwise Tensor Upload}
After completing APCP and KCCP, the master node transmits two pairwise coded tensor partitions to each worker
\(i\in\mathcal{Z}_n\), namely
\(\tilde{\mathbf{X}}_{\langle i,\beta_1\rangle}\) and
\(\tilde{\mathbf{K}}_{\langle i,\beta_2\rangle}\) with
\((\beta_1,\beta_2)\in\mathcal{Z}_2\times\mathcal{Z}_2\).

\subsubsection[Pairwise Tensor Convolution]{Pairwise Tensor Convolution}  
Upon reception, worker \(i\) performs pairwise tensor convolutions worker to get four output tensors \(\tilde{\mathbf{Y}}_{\langle i, \beta_3 \rangle} \in \mathbb{R}^{\frac{N}{k_B} \times \frac{H'}{k_A} \times W'}\) for \(\beta_3 \in \mathcal{Z}_4\), defined as follows:
\begin{equation}
\begin{aligned}
    &[\tilde{\mathbf{X}}_{\langle i, 0 \rangle}, \tilde{\mathbf{X}}_{\langle i, 1 \rangle}] * [\tilde{\mathbf{K}}_{\langle i, 0 \rangle}, \tilde{\mathbf{K}}_{\langle i, 1 \rangle}] \\
    &= [\tilde{\mathbf{Y}}_{\langle i, 0 \rangle}, \tilde{\mathbf{Y}}_{\langle i, 1 \rangle}, \tilde{\mathbf{Y}}_{\langle i, 2 \rangle}, \tilde{\mathbf{Y}}_{\langle i, 3 \rangle}]
\end{aligned}
\end{equation}
where
\begin{equation}
\tilde{\mathbf{Y}}_{\langle i, 2\beta_2 + \beta_1 \rangle} = \tilde{\mathbf{X}}_{\langle i, \beta_1 \rangle} * \tilde{\mathbf{K}}_{\langle i, \beta_2 \rangle},
\label{eq:indexing_relation}
\end{equation}
for \(i \in \mathcal{Z}_n\) and \((\beta_1, \beta_2) \in \mathcal{Z}_2 \times \mathcal{Z}_2\).

\subsubsection[Channel-wise concatenation]{Channel-wise concatenation}  
Each worker node \(i\) concatenates its convolution results along the channel dimension into a tensor \(\tilde{\mathbf{Y}}_i \in \mathbb{R}^{\frac{4N}{k_B} \times \frac{H'}{k_A} \times W'}\). This operation is represented as:
\begin{equation}\label{eq:output_order_optimized}
\tilde{\mathbf{Y}}_i = \operatorname{concat}_{\text{axis}=0}\left\{ \tilde{\mathbf{Y}}_{\langle i, 0 \rangle}, \tilde{\mathbf{Y}}_{\langle i, 1 \rangle}, \tilde{\mathbf{Y}}_{\langle i, 2 \rangle}, \tilde{\mathbf{Y}}_{\langle i, 3 \rangle} \right\},
\end{equation}
where \(\operatorname{concat}_{\text{axis}=0}\) denotes concatenation along the first dimension (channel dimension).

Upon completion of the channel-wise concatenation operation, the encoded tensor results \(\tilde{\mathbf{Y}}_i\) are immediately transmitted back to the master node for further processing.

Algorithm \ref{alg:FCDCC} presents a detailed overview of the master node's upload and computing phase following APCP and KCCP, as well as the distributed convolution performed by worker nodes.
\begin{algorithm}[H]  
\caption{Upload and Computing in FCDCC}  
\label{alg:FCDCC}  
\begin{algorithmic}[1]  
\renewcommand{\algorithmicrequire}{\textbf{Input:}}  
\renewcommand{\algorithmicensure}{\textbf{Output:}}  
\Require Encoded input tensor list $\tilde{\mathbf{X}}$, encoded filter tensor list $\tilde{\mathbf{K}}$, number of worker nodes $n$  
\Ensure Distributed convolution results $\{\tilde{\mathbf{Y}}_i\}_{i=0}^{n-1}$  

\Statex \textbf{Upload Data to Worker Nodes:}
\For{$i \gets 0$ \textbf{to} $n - 1$}
    \State Transmit tensor list $[\tilde{\mathbf{X}}_{\langle i, 0 \rangle}, \tilde{\mathbf{X}}_{\langle i, 1 \rangle}]$ to worker node $i$
    \State Transmit tensor list $[\tilde{\mathbf{K}}_{\langle i, 0 \rangle}, \tilde{\mathbf{K}}_{\langle i, 1 \rangle}]$ to worker node $i$
\EndFor

\Statex \textbf{Parallel Convolution Computation:}
\ForAll{$i \gets 0$ \textbf{to} $n - 1$} \textbf{in parallel on worker node $i$}
    \For{$\beta_1 \gets 0$ \textbf{to} $1$}  
        \For{$\beta_2 \gets 0$ \textbf{to} $1$}  
            \State $\tilde{\mathbf{Y}}_{\langle i, 2\beta_2 + \beta_1 \rangle} \gets \tilde{\mathbf{X}}_{\langle i, \beta_1 \rangle} \ast \tilde{\mathbf{K}}_{\langle i, \beta_2 \rangle}$  
        \EndFor  
    \EndFor  
    \State $\tilde{\mathbf{Y}}_i \gets \operatorname{concat}_{\text{axis}=0}\left(\tilde{\mathbf{Y}}_{\langle i, 0 \rangle}, \tilde{\mathbf{Y}}_{\langle i, 1 \rangle}, \tilde{\mathbf{Y}}_{\langle i, 2 \rangle}, \tilde{\mathbf{Y}}_{\langle i, 3 \rangle}\right)$  
\EndFor  

\State \Return $\{\tilde{\mathbf{Y}}_i\}_{i=0}^{n-1}$  
\end{algorithmic}  
\end{algorithm}
\subsection{Decoding and Merge Phase in FCDCC}

The decoding and merge phase of the FCDCC framework is designed to systematically recover the complete output tensor by processing encoded subresults obtained from the worker nodes. A key distinction from conventional CDC schemes lies in FCDCC's employment of specific tensor manipulation primitives, namely vectorization ($\operatorname{vec}(\cdot)$) and reshaping ($\operatorname{reshape}(\cdot)$), complemented by concatenation operations performed explicitly along the spatial and channel dimensions of the tensor. This recovery process is outlined as:
\subsubsection{Constructing the Joint Encoding Matrix}:  
The master node constructs the joint encoding matrix \(\mathbf{G}\) using the Kronecker product of the encoding matrices \(\mathbf{A}\) and \(\mathbf{B}\):
\begin{equation}\label{eq:kronecker}
\mathbf{G} = \mathbf{A} \otimes \mathbf{B} = \left[\mathbf{A}_0 \otimes \mathbf{B}_0 \,\, | \,\, \ldots \,\, | \,\, \mathbf{A}_{n-1} \otimes \mathbf{B}_{n-1}\right],
\end{equation}
where \(\mathbf{A}_i\) and \(\mathbf{B}_i\) are the \(i\)-th column blocks of \(\mathbf{A}\) and \(\mathbf{B}\), respectively.

\subsubsection{Forming the Decoding Matrix}  
    Upon receiving the earliest \(\delta = \frac{k_A k_B}{4}\) encoded convolution results from worker nodes, the master node forms the index set \(\mathcal{I} = \{i_1, i_2, \dots, i_{\delta}\}\) ($|\mathcal{I}| = \delta $) based on their corresponding worker node indices \(i\). The recovery matrix \(\mathbf{E} \in \mathbb{R}^{k_A k_B \times k_A k_B}\) is then constructed using the corresponding column blocks of \(\mathbf{G}\):
    \begin{equation}\label{eq:E}
        \mathbf{E} = \left[\mathbf{G}_{i_1} \,\, \mathbf{G}_{i_2} \,\, \cdots \,\, \mathbf{G}_{i_{\delta}}\right].
    \end{equation}
    
    The decoding matrix \(\mathbf{D}\) is then obtained by inverting the recovery matrix \(\mathbf{E}\):
    \begin{equation}
        \mathbf{D} = \mathbf{E}^{-1}.
    \end{equation}

    \subsubsection{Vectorizing and Concatenating Encoded Results}
    The master node collects the encoded convolution results \(\tilde{\mathbf{Y}}_i\) from worker nodes in the index set \(\mathcal{I}\). Each tensor block \(\tilde{\mathbf{Y}}_{\langle i, \beta_3 \rangle} \in \mathbb{R}^{\frac{N}{k_B} \times \frac{H'}{k_A} \times W'}\) in \(\tilde{\mathbf{Y}}_i\) is vectorized into a column vector \(\tilde{\mathbf{Y}}_{\langle i, \beta_3 \rangle, \text{vec}} \in \mathbb{R}^{\left(\frac{N}{k_B} \frac{H'}{k_A} W'\right) \times 1}\) using the \(\operatorname{vec}\) operation, which flattens the tensor in lexicographical order. Then, the vectorized tensors are concatenated along the column axis to form \(\tilde{\mathbf{Y}}_{\text{vec}} \in \mathbb{R}^{\left(\frac{N}{k_B} \frac{H'}{k_A} W'\right) \times k_A k_B}\):
    \begin{equation}\label{eq:vect_concatenated_tensor}
        \tilde{\mathbf{Y}}_{\text{vec}} = \operatorname{concat}_{\text{axis}=1}\left\{ \tilde{\mathbf{Y}}_{i, \text{vec}} \mid i \in \mathcal{I} \right\}.
    \end{equation}

    \subsubsection{Decoding the Vectorized Results}
    The decoded vector \(\mathbf{Y}_{\text{vec}}\) is obtained through matrix multiplication with the decoding matrix \(\mathbf{D}\):
    \begin{equation}\label{decodeeq}
        \mathbf{Y}_{\text{vec}} = \tilde{\mathbf{Y}}_{\text{vec}} \mathbf{D}.
    \end{equation}

     \subsubsection{Reshaping to Tensor Blocks} 
    The vector \(\mathbf{Y}_{\text{vec}}\) is partitioned and reshaped to reconstruct the individual tensor blocks \(\mathbf{Y}_{(u_A, u_B)} \in \mathbb{R}^{\frac{N}{k_B} \times \frac{H'}{k_A} \times W'}\), where \(u_A \in \mathcal{Z}_{k_A}\) and \(u_B \in \mathcal{Z}_{k_B}\). The mapping from column indices to tensor blocks is defined by:
    \begin{equation}
        i = u_B k_A + u_A.
    \end{equation}
    Each tensor block is obtained by:
    \begin{equation}
        \mathbf{Y}_{(u_A, u_B)} = \operatorname{reshape}\left(\mathbf{Y}_{\text{vec}}[:, i], \left(\frac{N}{k_B}, \frac{H'}{k_A}, W'\right)\right).
    \end{equation}
    The \(\operatorname{reshape}\) function converts the vectorized data \(\mathbf{Y}_{\text{vec}}[:, i]\) back into its original three-dimensional tensor form, thereby restoring the spatial and channel structure of the  output tensor. 

    \subsubsection{Assembling the Final Output Tensor}
    First, concatenate tensors along axis = 1 (height dimension \(\frac{H'}{k_A}\)), grouping blocks with the same \(u_B\) index but different \(u_A\) indices:
    \begin{equation}
    \scalebox{0.93}{$
    \mathbf{Y}_{u_B} = \operatorname{concat}_{\text{axis}=1}\left\{ \mathbf{Y}_{(i, u_B)} \mid \forall i \in \mathcal{Z}_{k_A} \right\}, \text{for}\ u_B \in \mathcal{Z}_{k_B}.
    $}
    \end{equation}
    
    Next, concatenate the intermediate tensors \(\mathbf{Y}_{u_B}\), now each of dimensions \(\mathbb{R}^{\frac{N}{k_B} \times H' \times W'}\), along axis = 0 (channel dimension \(\frac{N}{k_B}\)):
    \begin{equation}
    \mathbf{Y} = \operatorname{concat}_{\text{axis}=0}\left\{ \mathbf{Y}_{u_B} \mid \forall u_B \in \mathcal{Z}_{k_B} \right\}.
    \end{equation}

    This method ensures that all tensor blocks \(\mathbf{Y}_{(u_A, u_B)}\) are systematically reassembled into the final output tensor \(\mathbf{Y}\).

The entire decoding and merging process is summarized in Algorithm~\ref{alg:FCDCC_Decoding}.

\begin{algorithm}[H]
\caption{Decoding and Merge in FCDCC}
\label{alg:FCDCC_Decoding}
\begin{algorithmic}[1]
\Require Encoded results \(\{\tilde{\mathbf{Y}}_i\}_{i \in \mathcal{I}}\), joint encoding matrix \(\mathbf{G}\), index set \(\mathcal{I}\)
\Ensure Output tensor \(\mathbf{Y} \in \mathbb{R}^{N \times H' \times W'}\)

\State \textbf{Step 1:} Construct the recovery matrix \(\mathbf{E}\)
\State \(\mathbf{E} \gets [\mathbf{G}_{i_1} \,\, \mathbf{G}_{i_2} \,\, \cdots \,\, \mathbf{G}_{i_{\delta}}]\), where \(i_j \in \mathcal{I}\)

\State \textbf{Step 2:} Compute the decoding matrix \(\mathbf{D}\)
\State \(\mathbf{D} \gets \mathbf{E}^{-1}\)

\State \textbf{Step 3:} Vectorize and concatenate encoded results
\For{\(i \in \mathcal{I}\)}
    \State \(\tilde{\mathbf{Y}}_{i, \text{vec}} \gets \operatorname{vec}(\tilde{\mathbf{Y}}_{\langle i, \beta_3 \rangle})\) for all \(\beta_3\)
\EndFor
\State \(\tilde{\mathbf{Y}}_{\text{vec}} \gets \operatorname{concat}_{\text{axis}=1}\{\tilde{\mathbf{Y}}_{i, \text{vec}} \mid i \in \mathcal{I}\}\)

\State \textbf{Step 4:} Decode the vectorized results
\State \(\mathbf{Y}_{\text{vec}} \gets \tilde{\mathbf{Y}}_{\text{vec}} \mathbf{D}\)

\State \textbf{Step 5:} Reshape to tensor blocks
\For{\(u_B = 0\) to \(k_B - 1\)}
    \For{\(u_A = 0\) to \(k_A - 1\)}
        \State \(i \gets u_B k_A + u_A\)
        \State \(\mathbf{Y}_{(u_A, u_B)} \gets \operatorname{reshape}\left(\mathbf{Y}_{\text{vec}}[:, i], \left(\tfrac{N}{k_B}, \tfrac{H'}{k_A}, W'\right)\right)\)
    \EndFor
\EndFor

\State \textbf{Step 6:} Assemble the final output tensor
\For{\(u_B = 0\) to \(k_B - 1\)}
    \State \(\mathbf{Y}_{u_B} \gets \operatorname{concat}_{\text{axis}=1}\{\mathbf{Y}_{(u_A, u_B)} \mid u_A \in \mathcal{Z}_{k_A}\}\)
\EndFor
\State \(\mathbf{Y} \gets \operatorname{concat}_{\text{axis}=0}\{\mathbf{Y}_{u_B} \mid u_B \in \mathcal{Z}_{k_B}\}\)

\State \Return \(\mathbf{Y}\)
\end{algorithmic}
\end{algorithm}

\subsection{Optimization of Partitioning Parameters \texorpdfstring{\((k_A, k_B)\)}{(kA, kB)} in the FCDCC Framework}

In the FCDCC framework, the primary costs per worker node arise from communication, computation, and storage. The objective is to minimize the total cost \(U(k_A, k_B)\) per node by selecting optimal partitioning parameters \(k_A\) and \(k_B\), given a fixed total number of subtasks \(Q = k_A k_B\). We define the communication volumes \(V_{\text{comm\_up}}\) and \(V_{\text{comm\_down}}\), the storage volume \(V_{\text{store}}\), and the computation workload \(M_{\text{comp}}\) per node. Detailed complexity analysis is provided in Section V.

Given the cost per tensor entry \(\lambda_{\text{comm}}\) for communication, the upload and download communication costs are:
\begin{equation}
C_{\text{comm\_up}} = \lambda_{\text{comm}} V_{\text{comm\_up}} = \lambda_{\text{comm}} \frac{4C(H + 2p)(W + 2p)}{k_A},
\end{equation}
\begin{equation}
C_{\text{comm\_down}} = \lambda_{\text{comm}} V_{\text{comm\_down}} = \lambda_{\text{comm}} \frac{4N H' W'}{Q}.
\end{equation}

The total communication cost \(C_{\text{comm}}\) is then expressed as:
\begin{equation}
\label{eq:comm_cost_detailed}
\begin{aligned}
C_{\text{comm}} &= C_{\text{comm\_up}} + C_{\text{comm\_down}} = \lambda_{\text{comm}} (V_{\text{comm\_up}} + V_{\text{comm\_down}})\\
&= \lambda_{\text{comm}} \left( \frac{4C(H + 2p) (W + 2p)}{k_A} + \frac{4N H' W'}{Q} \right).
\end{aligned}
\end{equation}

The computation cost per node, with \(\lambda_{\text{comp}}\) as the cost per MAC operation, is:
\begin{equation}
C_{\text{comp}} = \lambda_{\text{comp}} M_{\text{comp}} = \lambda_{\text{comp}} \frac{4CN H W K_H K_W}{s^2 Q}.
\label{eq:comp_cost}
\end{equation}

The storage cost per node, influenced by \(\lambda_{\text{store}}\) per tensor entry, is:
\begin{equation}
C_{\text{store}} = \lambda_{\text{store}} V_{\text{store}} = \lambda_{\text{store}} \frac{2N C K_H K_W}{k_B}.
\label{eq:store_cost}
\end{equation}

The total cost per node is the sum of the above costs:
\begin{equation}
\label{eq:total_cost}
\begin{aligned}
U(k_A) &= C_{\text{comm}} + C_{\text{comp}} + C_{\text{store}} \\
&= \lambda_{\text{comm}} \left( \frac{4C(H + 2p)(W + 2p)}{k_A} + \frac{4N H' W'}{Q} \right) \\
&\quad + \lambda_{\text{comp}} \frac{4CN H W K_H K_W}{s^2 Q} + \lambda_{\text{store}} \frac{2N C K_H K_W k_A}{Q}.
\end{aligned}
\end{equation}

Letting constants \(a_1 = \lambda_{\text{store}} \frac{2N C K_H K_W}{Q}\), \(a_2 = \lambda_{\text{comm}} 4C(H + 2p)(W + 2p)\), and \(a_3 = \lambda_{\text{comm}} \frac{4N H' W'}{Q} + \lambda_{\text{comp}} \frac{4CN H W K_H K_W}{s^2 Q}\), we simplify \(U(k_A)\):
\begin{equation}
U(k_A) = a_1 k_A + a_2 \frac{1}{k_A} + a_3.
\end{equation}

\begin{lemma}[Convexity of the Cost Function]
The total cost function \(U(k_A)\) is strictly convex for \(k_A > 0\).
\end{lemma}
\begin{proof}
The second derivative of \(U(k_A)\) with respect to \(k_A\) is:
\begin{equation}
\frac{d^2U}{dk_A^2} = \frac{2a_2}{k_A^3} > 0 \quad \text{for all } k_A > 0,
\end{equation}
since \(a_2 > 0\). Therefore, \(U(k_A)\) is strictly convex.
\end{proof}

\begin{theorem}[Optimal Partitioning]
The optimal partitioning parameters are given by:
To find the optimal \(k_A\), we set the first derivative to zero:
\begin{equation}
\frac{dU}{dk_A} = a_1 - \frac{a_2}{k_A^2} = 0 \implies k_A^* = \sqrt{\frac{a_2}{a_1}}.
\end{equation}

Substituting back the expressions for \(a_1\) and \(a_2\):
\begin{equation}
k_A^* = \sqrt{\frac{2 \lambda_{\text{comm}} (H + 2p)(W + 2p) Q}{\lambda_{\text{store}} N K_H K_W}}.
\end{equation}
The corresponding \(k_B^*\) is:
\begin{equation}
k_B^* = \frac{Q}{k_A^*} = \sqrt{\frac{\lambda_{\text{store}} N K_H K_W Q}{2 \lambda_{\text{comm}} (H + 2p)(W + 2p)}}.
\end{equation}
\end{theorem}

Given the strict convexity of the function \(U_{k_A, k_B}\), the optimal integer solution is found by rounding \(k_A^*\) to the nearest value in the set \(\mathcal{S} = \{x \in \mathbb{Z}^+ \mid x = 1 \text{ or } x \equiv 0 \pmod{2}\}\) and evaluating the corresponding costs. If both \(k_A^*\) and \(k_B^* = \frac{Q}{k_A^*}\) belong to \(\mathcal{S}\), they represent the unique global minimum. Otherwise, compute \(U(k_A)\) for the nearest even integers to \(k_A^*\), select the value that minimizes the cost, and determine \(k_B^*\) by \(k_B^* = \frac{Q}{k_A^*}\).

\begin{lemma}[Asymptotic Behavior]
As \(Q \to \infty\), the optimal partitioning parameters scale as \(k_A^* = \Theta(\sqrt{Q})\) and \(k_B^* = \Theta(\sqrt{Q})\).
\end{lemma}
\begin{proof}
From the expressions for \(k_A^*\) and \(k_B^*\), both are proportional to \(\sqrt{Q}\). Thus, as \(Q\) increases:
\begin{equation}
k_A^* = \Theta(\sqrt{Q}), \quad k_B^* = \Theta(\sqrt{Q}).
\end{equation}
This balanced partitioning optimizes the trade-off between communication and storage costs.
\end{proof}
\section{Theoretical Analysis}
This section presents a comprehensive assessment of the FCDCC framework, focusing on its resilience to stragglers, condition number analysis, and complexity evaluations of the encoding, computation, communication, and decoding phases. All analyses are based on Multiply-Accumulate (MAC) operations and tensor entry counts. We also compare the FCDCC framework with existing model parallelism methods to evaluate its efficiency in distributed CNNs.
\subsection{Resilience and Condition Number Analysis}
In the FCDCC framework, the input tensor \(\mathbf{X}\) and filter tensor \(\mathbf{K}\) are partitioned into \(k_A\) and \(k_B\) segments, respectively. Utilizing Polynomial Codes within Circulant and Rotation Matrix Embeddings (CRME), the recovery threshold is \(\delta = \frac{k_A k_B}{4}\). With \(n\) worker nodes (\(n > \delta\)), the straggler resilience capacity is \(\gamma = n - \delta\). According to \cite{ramamoorthy2021numerically}, the worst-case condition number of the CRME Vandermonde decoding matrix \(\mathbf{D}\) in the FCDCC framework is bounded by \(\mathcal{O}(n^{\gamma + c_1})\), where \(c_1 \approx 5.5\). This represents optimal numerical stability in Polynomial Codes of existing CDC schemes as analyzed in \cite{fahim2021numerically,das2021efficient}.
\subsection{Encoding Complexity Analysis}
The encoding complexity in the FCDCC framework is critical for computational overhead. We analyze the complexity for encoding input and filter tensor partitions using two methods:
\begin{itemize}
    \item \textbf{Direct Linear Combination Method:} For tensor partitions \(\mathbf{X}'_i \in \mathbb{R}^{C \times \hat{H} \times (W + 2p)}\) with \(i \in \mathcal{Z}_{k_A}\), this method requires \(2n\) operations per tensor entry, resulting in a complexity of \(\mathcal{O}(2n C (H + 2p)(W + 2p))\), since \(\hat{H} \approx \frac{H+2p}{k_A}\).

    \item \textbf{Fast Polynomial Evaluation Method:} Treating encoding as polynomial evaluation of degree \(k_A - 1\) at \(2n\) points per tensor entry, the complexity reduces to \(\mathcal{O}\left(2n C (H + 2p) (W + 2p) \frac{\log^2 k_A \log \log k_A}{k_A}\right)\) using techniques from \cite{kedlaya2011fast}.
\end{itemize}

Similarly, for filter tensor partitions \(\mathbf{K}'_i \in \mathbb{R}^{\frac{N}{k_B} \times C \times K_H \times K_W}\) with \(i \in \mathcal{Z}_{k_B}\), the complexities are \(\mathcal{O}(2n N C K_H K_W)\) and \(\mathcal{O}\left(2n N C K_H K_W \frac{\log^2 k_B \log \log k_B}{k_B}\right)\) respectively.
\subsection{Computational, Communication and Storage Complexity}
Each worker node \(i\) performs four convolution operations between the coded input tensor partitions \(\tilde{\mathbf{X}}_{\langle i, 0 \rangle}\), \(\tilde{\mathbf{X}}_{\langle i, 1 \rangle}\) and the coded filter tensor partitions \(\tilde{\mathbf{K}}_{\langle i, 0 \rangle}\), \(\tilde{\mathbf{K}}_{\langle i, 1 \rangle}\). The computational complexity per node is \(\mathcal{O}(M_{\text{comp}})\), with $M_{\text{comp}} = \frac{4C N H W K_H K_W}{s^2 k_A k_B}$.

Assuming the filter tensor is uploaded and stored after the initial iteration, the storage complexity per node is \(\mathcal{O}(V_{\text{store}})\), with \(V_{\text{store}} = 2 \frac{N}{k_B} C K_H K_W\). For each inference iteration, the upload complexity for input tensor partitions per node is \(\mathcal{O}(V_{\text{comm\_up}})\), where \(V_{\text{comm\_up}} = 2C \hat{H} (W + 2p) \approx 2C \frac{H+2p}{k_A} (W + 2p)\). Each node produces four output partitions, leading to a download complexity per node of \(\mathcal{O}(V_{\text{comm\_down}})\), with \(V_{\text{comm\_down}} = 4 \frac{N H' W'}{k_A k_B}\).

\subsection{Decoding Complexity Analysis}
In the decoding phase, operations such as tensor vectorization and unvectorization, efficiently handled by contiguous memory storage, contribute negligibly to the overall complexity compared to other terms. A naive matrix inversion approach for \(\mathbf{D} \in \mathbb{R}^{k_A k_B \times k_A k_B}\)  would typically exhibit a complexity of \(\Oh((k_A k_B)^3)\). Subsequent recovery of the output tensor, as detailed in Eq.~\eqref{decodeeq}, incurs an additional complexity of \(\Oh(k_A k_B N H' W')\).

Similar to the encoding, the application of fast polynomial evaluation techniques can reduce the matrix multiplication complexity component in decoding to \(\Oh(NH'W'\log^2 (k_A k_B) \log \log (k_A k_B))\). Furthermore, by leveraging Vandermonde-like structures for inverting matrices, the complexity of the matrix inversion step can be reduced to \(\Oh((k_A k_B)^2)\). Therefore, the optimized overall decoding complexity is \(\Oh\left(k_A^2 k_B^2 + NH'W'\log^2 (k_A k_B) \log \log (k_A k_B)\right)\).

\subsection{Theoretical Analysis of the overhead impact}\label{subsec:overhead}

The overhead intrinsic to the FCDCC framework, relative to conventional distributed computing paradigms, is primarily concentrated in the encoding and decoding phases. Within this framework, we assume the filter tensor is pre-encoded and resident at each worker node. Consequently, the aggregate overhead encompasses three principal components: (i) encoding of the input tensor, (ii) inversion of the recovery matrix, and (iii) decoding of the output tensor. Let \(Q = k_A k_B\), then the per-node workload for the convolution task is \(\Oh\left(\frac{4 C N H W K_H K_W}{s^2 Q}\right)\). It is established that parameters \(p, s, K_H, K_W\) are typically substantially smaller than \(C, N, H, W, n\). This leads to the approximations \((H+2p)(W+2p) \approx HW\) and \(H'W' \approx HW/s^2\). These approximations, along with the common assumption that \(K_H K_W/s^2 = \Oh(1)\) under typical parameter scaling, are utilized in the subsequent simplifications unless explicitly stated otherwise.

In the worst case, devoid of optimizations for encoding/decoding and matrix inversion, the overhead is expressed as \(\Oh\left(n C H W + Q^3 + Q N \frac{HW}{s^2}\right)\). This overhead becomes non-negligible when any of its constituent parts scale comparably to the per-node workload theoretically. Specifically, this occurs in the following scenarios:
\begin{itemize}
    \item [(i)] The input encoding component becomes significant if \(Q = \OmegaBig\left(\frac{N K_H K_W}{n s^2}\right)\), simplifying to \(Q = \OmegaBig\left(\frac{4N}{n}\right)\) (leveraging \(K_H K_W/s^2 = \Oh(1)\)).
    \item [(ii)] The matrix inversion component becomes significant if \(Q = \OmegaBig\left(\left(\frac{4 C N H W K_H K_W}{s^2}\right)^{1/4}\right)\), simplifying to \(Q = \OmegaBig\left((4 C N H W)^{1/4}\right)\).
    \item [(iii)] The output decoding component becomes significant if \(Q = \OmegaBig\left(2\sqrt{C K_H K_W}\right)\), simplifying to \(Q = \OmegaBig\left(2\sqrt{C}\right)\).
\end{itemize}

Conversely, through the application of fast polynomial evaluation for encoding/decoding and accelerated algorithms for matrix inversion, the refined overall overhead is \(\Oh\left(n C H W \frac{\log^2 k_A \log\log k_A}{k_A} + Q^2 + N \frac{HW}{s^2}\log^2 Q \log\log Q\right)\). Under these optimized conditions, the overhead's significance is dictated by the following scaling relationships for \(Q\):
\begin{itemize}
    \item [(i)] For the optimized encoding component: The condition is \(Q = \OmegaBig\left(\frac{4 N K_H K_W k_A}{n s^2 \log^2 k_A \log\log k_A}\right)\), which approximates to \(Q = \OmegaBig\left(\frac{4 N k_A}{n \text{polylog}(k_A)}\right)\). Considering polylogarithmic factors are sub-dominant in these regimes, this can be represented as \(Q = \OmegaBig\left(\frac{4Nk_A}{n}\right)\).
    \item [(ii)] For the optimized matrix inversion component: \(Q = \OmegaBig\left(\left(\frac{4 C N H W K_H K_W}{s^2}\right)^{1/3}\right)\), simplifying to \(Q = \OmegaBig\left((4 C N H W)^{1/3}\right)\).
    \item [(iii)] For the optimized decoding component: We get \(Q \log^2 Q \log\log Q = \OmegaBig( 4 C K_H K_W)\). This simplifies to \(Q \text{polylog}(Q) = \OmegaBig(4 C)\), often further approximated to \(Q  = \OmegaBig(4 C)\).
\end{itemize}
These conditions delineate the regimes where the respective overhead components become critical factors in the overall performance of the FCDCC framework theoretically.

\subsection{Comparison with other Model Parallelism Schemes}
Mainstream model parallelism strategies for ConvLs include:
\begin{itemize}
    \item \textbf{Spatial Partitioning:} Dividing along spatial dimensions \((H, W)\) \cite{mao2017modnn}.
    \item \textbf{Output Channel Partitioning:} Dividing along the output channel dimension \((N)\) \cite{hadidi2020toward}.
    \item \textbf{Input Channel Partitioning:} Dividing along the input channel dimension \((C)\) \cite{jia2018exploring}.
\end{itemize}
Table~\ref{tab:model_parallelism} summarizes key parameters of these methods (assuming \(p=0\)). The proposed FCDCC framework integrates the advantages of both spatial and output channel partitioning strategies while avoiding the merge operation complexity associated with input channel partitioning. 

Specifically, when \(k_A = 1\), FCDCC corresponds to the spatial partitioning approach; when \(k_B = 1\), it aligns with the output channel partitioning approach; and when \(k_Ak_B = 1\), it aligns with the baseline approach.
\begin{table*}[!t]
\renewcommand{\arraystretch}{1.3}
\caption{Comparison of Model Parallelism Methods for Convolutional Layers}
\label{tab:model_parallelism}
\centering
\resizebox{\textwidth}{!}{%
\begin{tabular}{|l|c|c|c|c|c|c|c|}
\hline
\textbf{Method} & \textbf{Division Factor} & \textbf{Nodes} & \textbf{Input Tensor} & \textbf{Filter Tensor} & \textbf{Output Tensor} & \textbf{Communication} & \textbf{Merge Operation} \\
\hline
Baseline & N/A & 1 & $C \times H \times W$ & $N \times C \times K_H \times K_W$ & $N \times H' \times W'$ & $C H W + N H' W'$ & N/A \\
\hline
Spatial Partitioning & $k_A$ & $k_A$ & $C \times \hat{H} \times W$ & $N \times C \times K_H \times K_W$ & $N \times \frac{H'}{k_A} \times W'$ & $C \hat{H} W + N \frac{H'}{k_A} W'$ & Concatenation \\
\hline
Output Channel Partitioning & $k_B$ & $k_B$ & $C \times H \times W$ & $\frac{N}{k_B} \times C \times K_H \times K_W$ & $\frac{N}{k_B} \times H' \times W'$ & $C H W + \frac{N}{k_B} H' W'$ & Concatenation \\
\hline
Input Channel Partitioning & $k_C$ & $k_C$ & $\frac{C}{k_C} \times H \times W$ & $N \times \frac{C}{k_C} \times K_H \times K_W$ & $N \times H' \times W'$ & $\frac{C}{k_C} H W + N H' W'$ & Summation \\
\hline
FCDCC & $k_A, k_B$ & $\frac{k_A k_B}{4}$ & $C \times \hat{H} \times W$ & $\frac{N}{k_B} \times C \times K_H \times K_W$ & $\frac{N}{k_B} \times \frac{H'}{k_A} \times W'$ & $C \hat{H} W + \frac{N}{k_B} \frac{H'}{k_A} W'$ & Concatenation \\
\hline
\end{tabular}%
}
\end{table*}
\section{Experiments}
This section presents the experimental evaluation of the proposed FCDCC framework, conducted on Amazon EC2. The experiments focus on the inference of a single batch across various ConvLs of LeNet, AlexNet, and VGGNet models. Performance was assessed based on time efficiency, resilience to stragglers, and numerical stability.
\subsection{Experimental Setup}
The experiments were conducted using \texttt{mpi4py} to facilitate distributed computing. A single master node managed the encoding process and distributed the encoded tensor partitions to worker nodes, which performed convolution computations independently and returned the results asynchronously. The convolution function was a basic, unoptimized implementation in PyTorch 2.4.0 (CPU). To simulate straggling effects, artificial delays were introduced using the \texttt{sleep()} function, and worker node availability was randomized using \texttt{random.random()}. Except for the naive scheme in Experiment 1, all experiments utilized \texttt{t2.micro} instances (1 vCPU, 1 GiB memory).
\subsection{Experimental Results}
\subsubsection{Experiment 1: Performance Comparison of FCDCC and Naive Scheme}
This experiment compares the performance of the FCDCC framework with a naive scheme across various CNN architectures. The naive scheme was executed on a single \texttt{i3n.xlarge} instance (4 vCPUs, 32 GiB memory), while the FCDCC scheme utilized nineteen \texttt{t2.micro} instances, configured with \(n = 18\) worker nodes and one master node, and parameters \(\delta = 16\), \(\gamma = 2\).

The Mean Squared Error (MSE) between the output tensors of the naive and FCDCC schemes was calculated as:
\begin{equation}
\text{MSE} = \frac{1}{N H' W'} \sum_{n=0}^{N-1} \sum_{h=0}^{H'-1} \sum_{w=0}^{W'-1} \left( \mathbf{P}_{n,h,w} - \mathbf{T}_{n,h,w} \right)^2,
\end{equation}
where \(\mathbf{P}\) and \(\mathbf{T}\) are the output tensors from the naive and FCDCC schemes, respectively, with dimensions \(\mathbb{R}^{N \times H' \times W'}\).

The ConvLs of LeNet-5, AlexNet, and VGGNet were evaluated using the FCDCC configuration \((k_A, k_B) = (2, 32)\). The results are summarized in Table~\ref{tab:exp1_results}.
\begin{table}[!t]
\renewcommand{\arraystretch}{1.05}
\setlength{\tabcolsep}{2pt} % 紧凑列距；如仍挤可调为 1.5pt
\caption{Performance Comparison of FCDCC and Naive Scheme Across CNN Architectures}
\label{tab:exp1_results}
\centering
\footnotesize % 如仍挤可改为 \scriptsize
\begin{tabular*}{\columnwidth}{@{\extracolsep{\fill}}%
  ll
  S[table-format=2.3]     % Naive (s)
  S[table-format=2.3]     % FCDCC (s)
  S[table-format=1.2e-2]  % MSE (scientific)
  S[table-format=2.3]     % Decode (ms)
@{}}
\toprule
\multirow{2}{*}{\textbf{Model}} & \multirow{2}{*}{\textbf{Layer}} &
\multicolumn{4}{c}{\textbf{Performance Metrics}} \\
\cmidrule(l){3-6}
& &
\multicolumn{1}{c}{\begin{tabular}[t]{@{}c@{}}\textbf{Naive}\\\textbf{(s)}\end{tabular}} &
\multicolumn{1}{c}{\begin{tabular}[t]{@{}c@{}}\textbf{FCDCC}\\\textbf{(s)}\end{tabular}} &
\multicolumn{1}{c}{\begin{tabular}[t]{@{}c@{}}\textbf{MSE}\end{tabular}} &
\multicolumn{1}{c}{\begin{tabular}[t]{@{}c@{}}\textbf{Decode}\\\textbf{(ms)}\end{tabular}} \\
\midrule
\multirow{2}{*}{LeNet-5}
& Conv1 & 0.256 & 0.016 & 1.10e-30 & 0.275 \\
& Conv2 & 0.404 & 0.017 & 3.57e-29 & 0.298 \\
\midrule
\multirow{5}{*}{AlexNet}
& Conv1 & 10.392 & 0.940 & 4.28e-28 & 2.496 \\
& Conv2 & 6.851  & 0.614 & 6.71e-28 & 1.564 \\
& Conv3 & 3.820  & 0.210 & 3.92e-27 & 0.434 \\
& Conv4 & 4.291  & 0.308 & 5.60e-27 & 0.566 \\
& Conv5 & 2.933  & 0.205 & 3.89e-27 & 0.455 \\
\midrule
\multirow{9}{*}{VGGNet}
& Conv1\_1       & \multicolumn{1}{c}{\text{N/A}} & 10.334 & \multicolumn{1}{c}{\text{N/A}} & 23.414 \\
& Conv1\_2       & \multicolumn{1}{c}{\text{N/A}} & 11.885 & \multicolumn{1}{c}{\text{N/A}} & 25.414 \\
& Conv2\_1       & 60.097 & 5.554 & 2.87e-28 & 9.763 \\
& Conv2\_2       & 63.123 & 6.351 & 4.97e-28 & 8.234 \\
& Conv3\_1       & 31.443 & 3.005 & 2.33e-27 & 6.226 \\
& Conv3\_2/3     & 37.277 & 3.659 & 3.67e-27 & 9.951 \\
& Conv4\_1       & 18.696 & 1.664 & 6.41e-27 & 1.951 \\
& Conv4\_2/3     & 24.538 & 2.222 & 1.01e-26 & 1.978 \\
& Conv5\_1/2/3   & 6.737  & 0.460 & 8.07e-27 & 0.683 \\
\bottomrule
\end{tabular*}
\end{table}
As shown in Table~\ref{tab:exp1_results}, the FCDCC scheme significantly reduces computation times across various CNN architectures and layers. For LeNet-5, reductions of 93.75\% and 95.80\% were achieved for the first and second ConvLs, respectively. In AlexNet, computation times were reduced by over 90\% across all ConvLs. For VGGNet, characterized by its depth and extended computation durations, the FCDCC scheme enabled substantial reductions. Initial layers, which could not be completed under the naive scheme due to resource constraints, were processed efficiently using FCDCC. Subsequent layers experienced computation time reductions ranging from approximately 90\% to 93\%. The MSE remained negligible, ranging from \(10^{-30}\) to \(10^{-26}\). Also, the decoding overhead amounts to only 0.1\% to 1.8\% of the worker-side computation time in FCDCC, confirming that master-side overhead is negligible compared to the actual computation at the tested scale. Moreover, we observe that this overhead ratio grows monotonically as the partition count \(Q\) approaches the dominance threshold predicted in Sec.~V.E, a trend that precisely corroborates our theoretical analysis.

These results demonstrate the efficiency of the FCDCC scheme in accelerating convolutional layer computations, particularly as the spatial dimensions \((H, W)\) of the input tensor and the number of output channels \(N\) increase. The FCDCC framework effectively distributes workloads across multiple nodes using APCP and KCCP.
\vspace{-6pt}
\begin{figure}[htbp]
  \centering
  \includegraphics[width=\columnwidth]{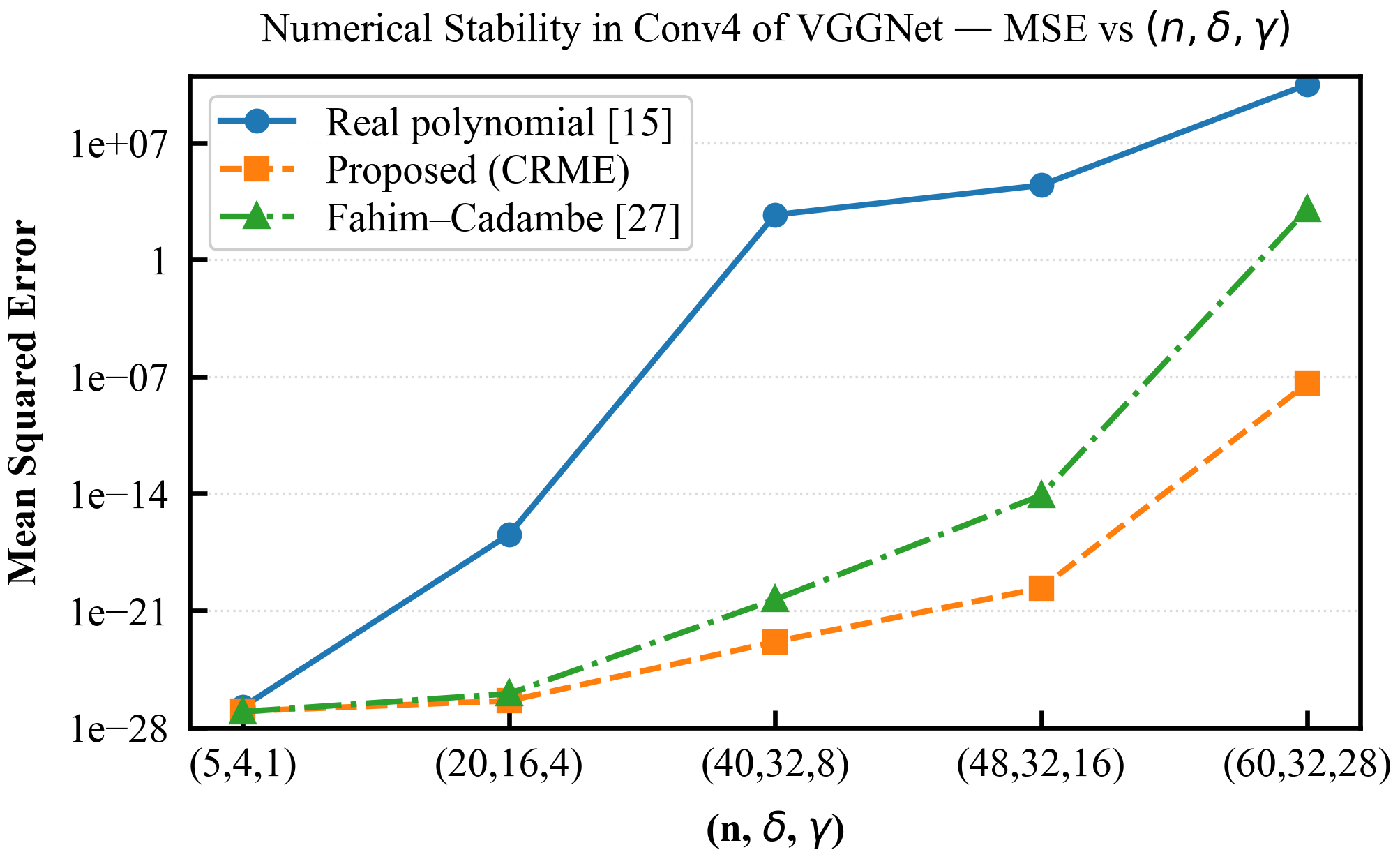}
  \vspace{-20pt} % ← 拉近图与caption
  \caption{Comparison of MSE for Different Numerically Stable CDC Methods in Conv4 of VGGNet}
  \label{EXMSE}
\end{figure}
\vspace{-6pt}
\subsubsection{Experiment 2: MSE and Condition Number Comparison of Numerically Stable CDC Schemes}
In this experiment, we evaluated the numerical stability of various CDC schemes using instances configured with \((n,\delta,\gamma)\in\{(5,4,1),(20,16,4),(40,32,8),(48,32,16),(60,32,28)\}\) for Conv4 of VGGNet. To the best of our knowledge, these numerically stable CDC schemes have not been previously extended to tensor convolution. The MSE and condition number were assessed under the considered schemes. As shown in Fig.~\ref{EXMSE} and Fig.~\ref{EXMSE2}, the proposed FCDCC framework based on CRME consistently achieves the lowest MSE and condition number, thereby exhibiting the highest resilience against numerical instability. Notably, the Real polynomial approach becomes numerically unstable at $(40,32,8)$, while the Fahim–Cadambe scheme demonstrates significant instability at $(60,32,28)$. 

\vspace{-12pt}
\begin{figure}[htbp]
  \centering
  \includegraphics[width=\columnwidth]{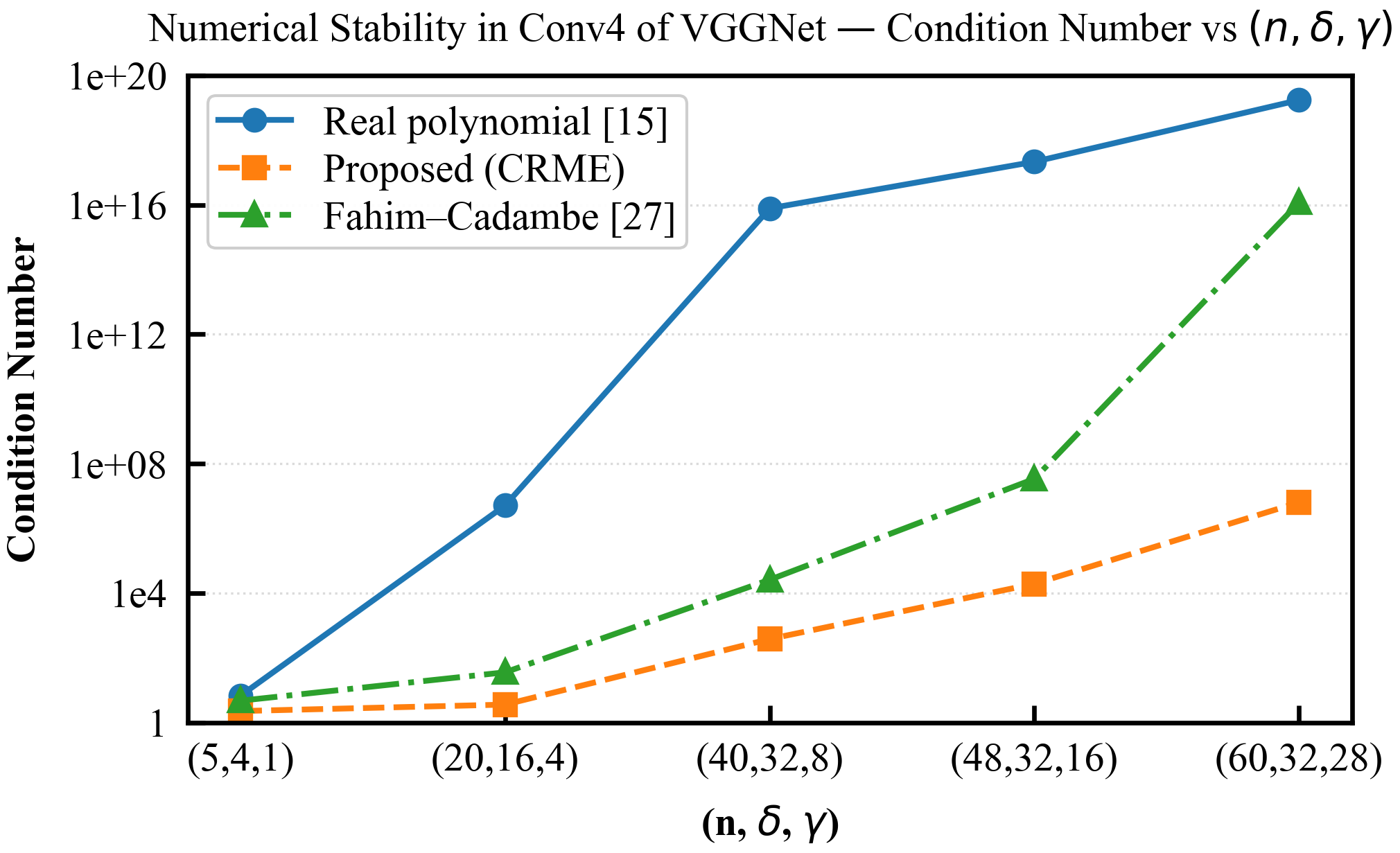}
  \vspace{-20pt}
  \caption{Comparison of Condition Number for Different Numerically Stable CDC Methods in Conv4 of VGGNet}
  \label{EXMSE2}
\end{figure}
\vspace{-8pt}

These results highlight that the FCDCC scheme markedly outperforms existing alternatives in terms of numerical stability for distributed tensor convolution, and therefore offers greater scalability before encountering instability issues.

\subsubsection{Experiment 3: Impact of \texorpdfstring{$n$}{n} and \texorpdfstring{$\delta$}{delta} on Average Computation Time}
This experiment evaluates the scalability of the FCDCC scheme by varying the number of worker nodes \(n\) and the recovery threshold \(\delta\), using the ConvLs of AlexNet for performance assessment. We set \(\gamma = 4\), with \(n\) ranging from 8 to 36 and \(\delta\) from 4 to 32.
\vspace{-6pt}
\begin{figure}[htbp]
  \centering
  \includegraphics[width=\columnwidth]{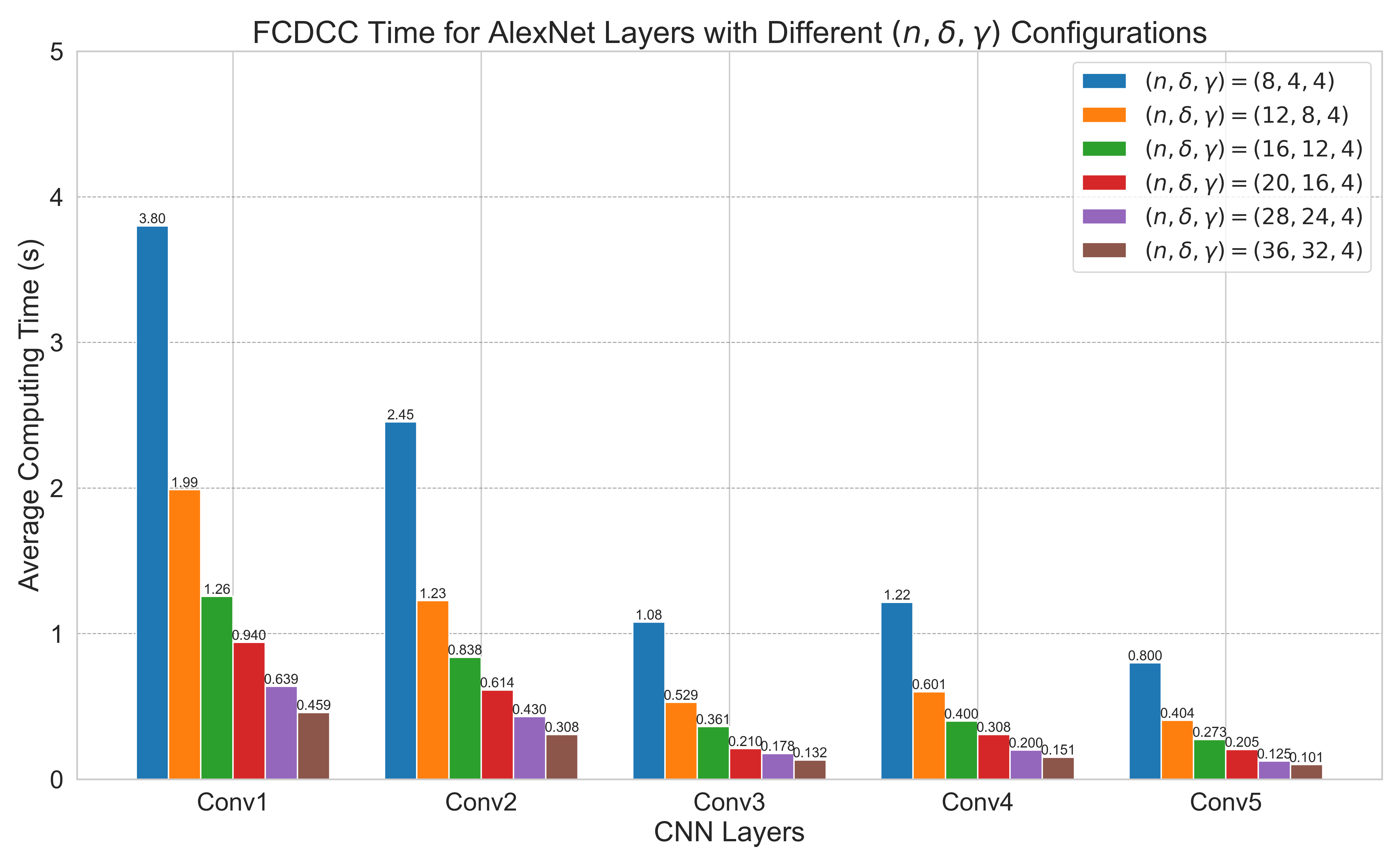}
  \vspace{-20pt} % ← 拉近图与caption
  \caption{Average Computation Time in FCDCC with Different \(n\) and \(\delta\)}
  \label{EX1}
\end{figure}
\vspace{-6pt}
As shown in Fig.~\ref{EX1}, the average computation time decreases as \(n\) and \(\delta\) increase, demonstrating the scalability of the FCDCC scheme. The results indicate that FCDCC efficiently utilizes additional computational resources, significantly reducing the average computation time for ConvLs.

\subsubsection{Experiment 4: Robustness Under Diverse Straggler Conditions}
This experiment assesses the robustness of the FCDCC framework under varying straggler conditions. Average computational latency was measured using thirty-three instances with \(n = 32\) worker nodes, \(\delta = 24\), and \(\gamma = 8\). We recorded average computation times for the ConvLs of AlexNet, varying the number of stragglers from 0 to 12, with artificial delays of 1 and 2 seconds.
\vspace{-12pt}
\begin{figure}[htbp]
  \centering
  \includegraphics[width=\columnwidth]{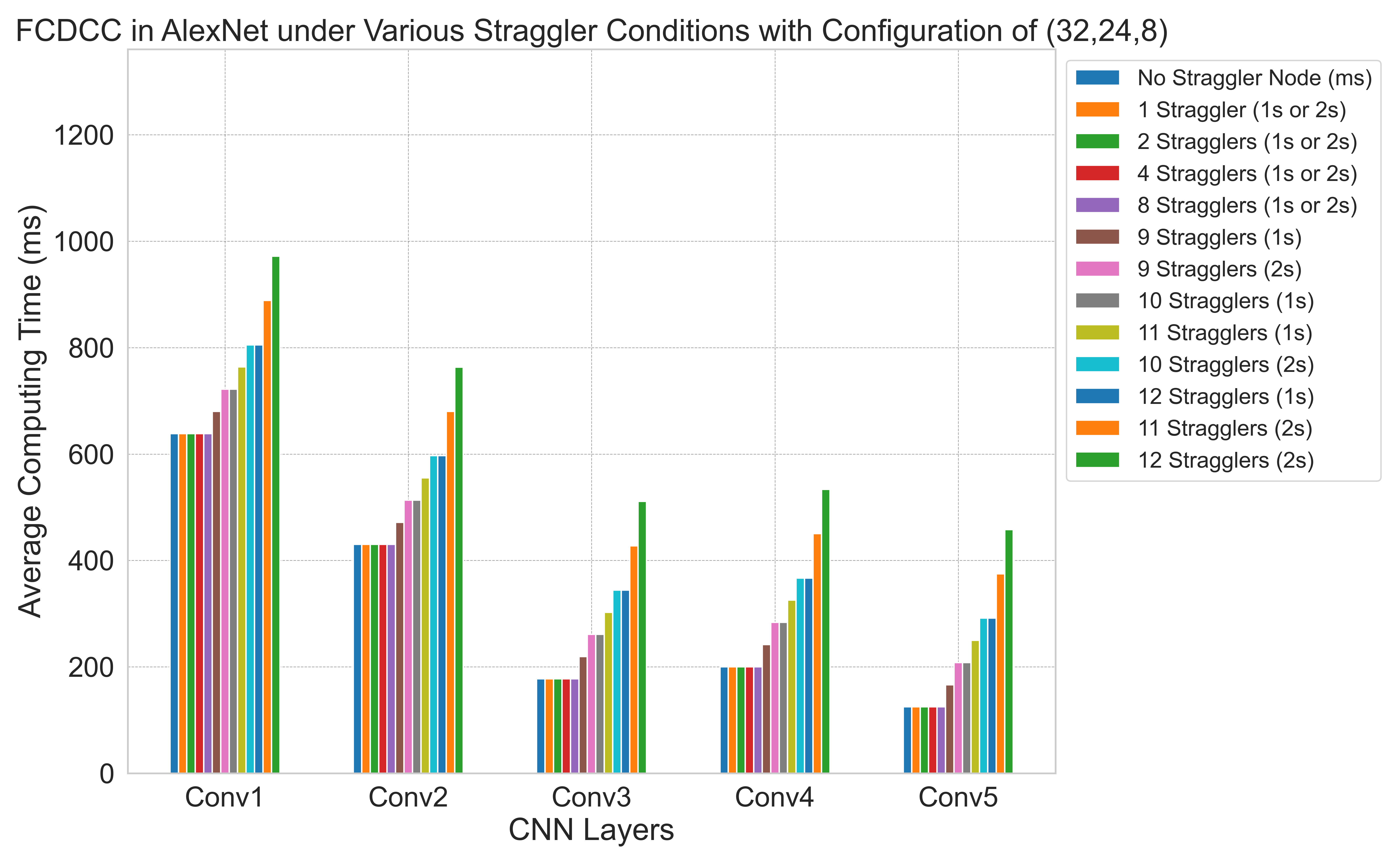}
  \vspace{-20pt} % ← 拉近图与caption
  \caption{Average Computation Time in FCDCC Under Varying Straggler Numbers and Delay Durations}
  \label{EX5}
\end{figure}
\vspace{-10pt}

Fig.~\ref{EX5} illustrates that the FCDCC scheme effectively mitigates the impact of stragglers as long as their number does not exceed the straggler tolerance capacity \(\gamma\), regardless of delay duration. Performance degradation becomes noticeable only when the number of stragglers surpasses \(\gamma\), highlighting the robustness of the FCDCC framework in distributed environments.
\subsubsection{Experiment 5: Optimization of Partition Configurations \texorpdfstring{\((k_A, k_B)\)}{(kA, kB)} for Various CNNs}
In this experiment, we optimize the partition configurations \((k_A, k_B)\) for ConvLs in LeNet, AlexNet, and VGGNet to minimize the total cost per node \(U_{k_A, k_B}\). Since the computation cost \(C_{\text{comp}}\) remains constant for a given \(Q\), we focus on minimizing the combined communication cost \(C_{\text{comm}}\) and storage cost \(C_{\text{store}}\), setting \(\lambda_{\text{comp}} = 0\).

To reflect realistic cost structures in large-scale distributed computing environments, we adopt cost coefficients \(\lambda_{\text{store}} = 0.023\) and \(\lambda_{\text{comm}} = 0.09\), based on AWS S3 pricing ratios for storage and communication per GB \cite{aws_s3_pricing}. We evaluate partition configurations for \(Q = 16, 32, 64\).

Fig.~\ref{fig:AlexC_1_2} illustrates the optimization landscape of \(U(k_A, k_B)\) for the first two ConvLs of AlexNet with \(Q = 32\). Discrete feasible points are highlighted in blue, and the optimal configurations \((k_A^*, k_B^*)\) are marked in red. The dashed line represents the constraint \(k_A k_B = Q\).
\begin{figure}[!t]
  \centering
  \subfloat[Layer 1]{\includegraphics[width=0.48\columnwidth]{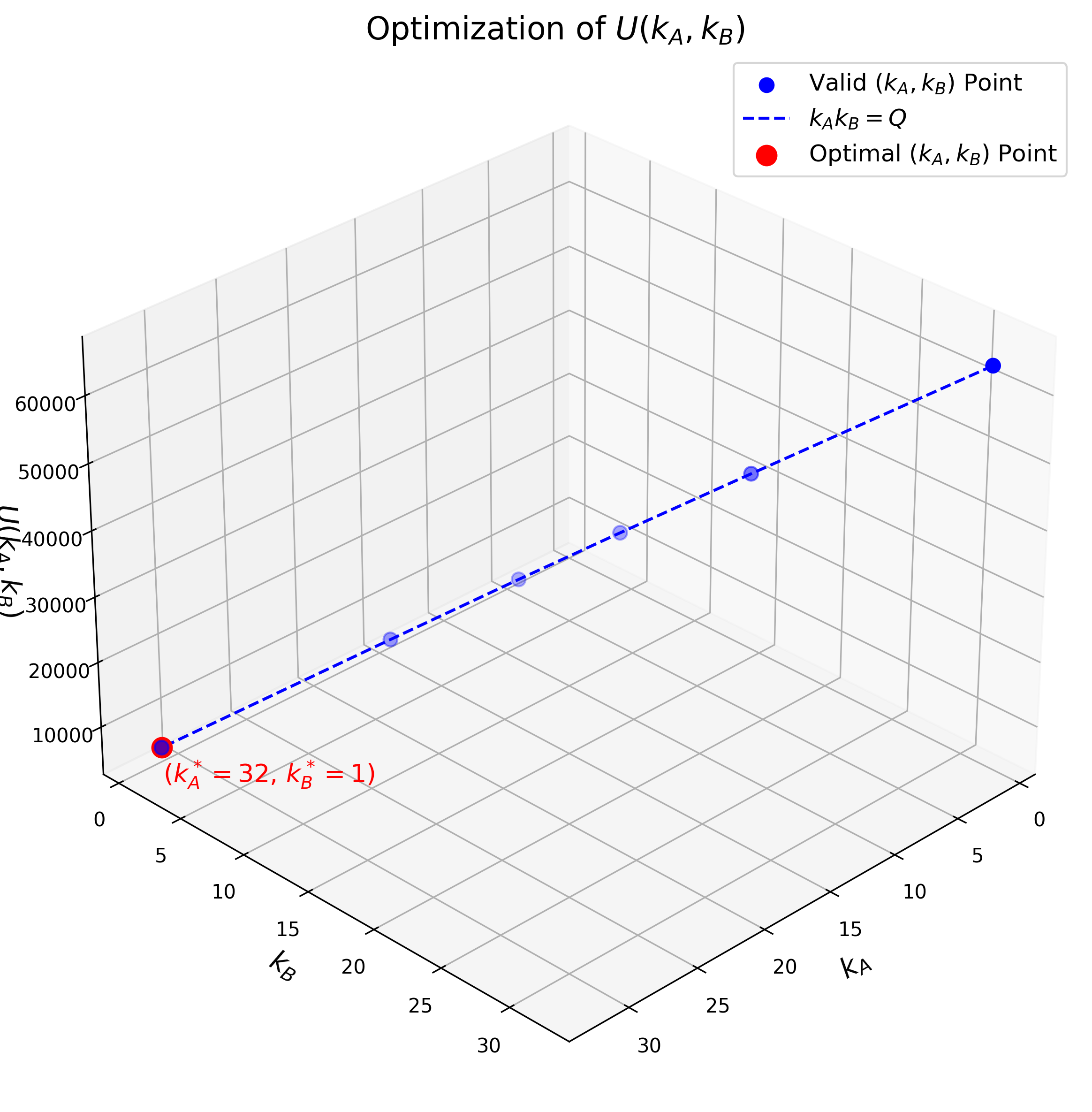}%
  \label{fig:AlexC_1}}
  \hfil
  \subfloat[Layer 2]{\includegraphics[width=0.48\columnwidth]{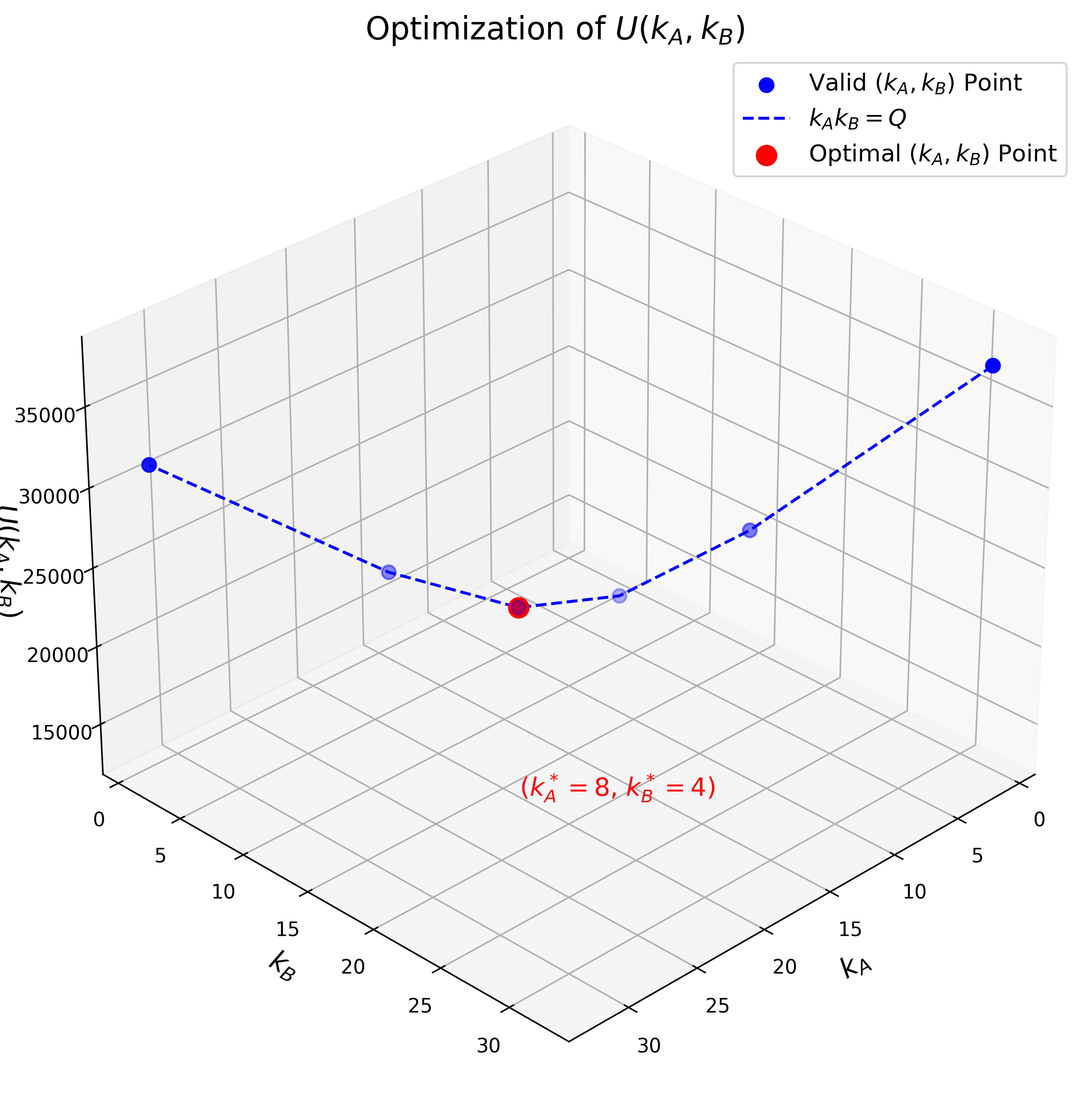}%
  \label{fig:AlexC_2}}
  \caption{Optimization of \(U(k_A, k_B)\) for the first two ConvLs of AlexNet with \(Q = 32\) and \(\lambda_{\text{comp}} = 0\).}
  \label{fig:AlexC_1_2}
\end{figure}
Table~\ref{table:cnn_optimization} summarizes the optimized partition configurations for the ConvLs in LeNet, AlexNet, and VGGNet. 

In early layers, the input tensors have larger spatial dimensions \((H, W)\) and fewer output channels \(N\), making communication cost \(C_{\text{comm}}\) the dominant factor. Thus, larger \(k_A\) and smaller \(k_B\) minimize \(C_{\text{comm}}\). In deeper layers, as \(N\) increases and \((H, W)\) decrease, storage cost \(C_{\text{store}}\) becomes more significant, favoring configurations with smaller \(k_A\) and larger \(k_B\). Additionally, the optimal values of \(k_A\) and \(k_B\) increase proportionally with \(Q\), reflecting the need for balanced partitioning as the total number of subtasks grows.

These results demonstrate that layer-specific partitioning effectively optimizes the trade-off between communication and storage costs in distributed CNN inference, enhancing overall cost efficiency.
\vspace{-6pt}
\begin{table}[!t]
\renewcommand{\arraystretch}{1.3}
\caption{Optimized \((k_A, k_B)\) Configurations for Various CNN Architectures}
\label{table:cnn_optimization}
\centering
\begin{tabular}{|l|c|c|c|c|c|c|}
\hline
\textbf{CNN Model} & \textbf{Q} & \textbf{Conv1} & \textbf{Conv2} & \textbf{Conv3} & \textbf{Conv4} & \textbf{Conv5} \\
\hline\hline
\multirow{3}{*}{LeNet-5} & 16 & (16, 1) & (8, 2) & -- & -- & -- \\
& 32 & (32, 1) & (16, 2) & -- & -- & -- \\
& 64 & (32, 2) & (16, 4) & -- & -- & -- \\
\hline
\multirow{3}{*}{AlexNet} & 16 & (16, 1) & (4, 4) & (2, 8) & (2, 8) & (2, 8) \\
& 32 & (32, 1) & (8, 4) & (2, 16) & (2, 16) & (4, 8) \\
& 64 & (32, 2) & (8, 8) & (4, 16) & (4, 16) & (4, 16) \\
\hline
\multirow{3}{*}{VGGNet} & 16 & (16, 1) & (16, 1) & (16, 1) & (4, 4) & (2, 8) \\
& 32 & (32, 1) & (32, 1) & (16, 2) & (8, 4) & (4, 8) \\
& 64 & (32, 2) & (32, 2) & (32, 2) & (8, 8) & (4, 16) \\
\hline
\end{tabular}
\end{table}
\vspace{-6pt}
\subsection{Summary of Experimental Results}
Experiments on LeNet-5, AlexNet, and VGGNet demonstrate that the proposed FCDCC framework significantly improves computational efficiency in distributed CNN inference, achieving computation time reductions exceeding 90\% compared to the naive scheme, while maintaining high numerical stability with negligible MSE ranging from \(10^{-30}\) to \(10^{-26}\). Also, FCDCC scales effectively, with computation time decreasing as the number of worker nodes \(n\) and the recovery threshold \(\delta\) increase, and tolerates up to \(\gamma\) stragglers without degradation. Optimizing partition configurations \((k_A, k_B)\) improves cost efficiency by balancing communication and storage costs. These results confirm the efficacy of FCDCC in distributed CNN inference.

\section{Conclusion and Future Work}  
In this paper, we proposed FCDCC framework, which integrates NSCTC with APCP and KCCP schemes. The FCDCC framework enhances numerical stability, system resilience, computational efficiency, and cost-effectiveness in distributed CNNs. Our theoretical analysis and extensive experimental results on networks such as LeNet-5, AlexNet, and VGGNet demonstrate that FCDCC significantly improves computational performance compared to traditional uncoded and existing coded schemes. Future work includes refining the coding mechanisms, extending the CDC scheme to support pooling layers and nonlinear activation functions, and enhancing privacy protection to safeguard against colluding and malicious worker nodes. 
\bibliographystyle{IEEEtran}  
\bibliography{reference}
\end{document}